\documentclass[journal]{IEEEtran}

\usepackage{amsmath,amssymb,amsfonts}
\usepackage{mathtools}
\usepackage{algorithmic}
\usepackage{algorithm}
\usepackage{graphicx}
\usepackage{textcomp}
\usepackage{xcolor}
\usepackage{bm}
\usepackage{booktabs}
\usepackage{cite}
\usepackage{url}
\usepackage{amsthm}
\newtheorem{theorem}{Theorem}
\newtheorem{lemma}{Lemma}
\newtheorem{proposition}{Proposition}
\newtheorem{corollary}{Corollary}
\newtheorem{assumption}{Assumption}
\newtheorem{remark}{Remark}
\newtheorem{definition}{Definition}

\DeclareMathOperator*{\argmax}{arg\,max}
\newcommand{\R}{\mathbb{R}}
\newcommand{\E}{\mathbb{E}}
\newcommand{\Prob}{\mathbb{P}}
\newcommand{\Nrm}{\mathcal{N}}
\newcommand{\wout}{w^{\mathrm{out}}}
\newcommand{\win}{w^{\mathrm{in}}}
\newcommand{\Dout}{D^{\mathrm{out}}}
\newcommand{\Din}{D^{\mathrm{in}}}
\newcommand{\onev}{\bm{1}}
\newcommand{\convas}{\stackrel{\text{a.s.}}{\longrightarrow}}

\newcommand{\includegraphicssafe}[2][]{%
  \IfFileExists{#2}{\includegraphics[#1]{#2}}{%
    \fbox{\parbox{0.92\columnwidth}{\centering Missing figure file:\\\texttt{#2}}}%
  }%
}

\begin{document}

\title{Spatial Dependence in Directed
Preferential-Attachment Networks}

\author{Zihan Li
        and~Tiandong Wang%
\thanks{
This work was supported by the National Natural Science Foundation of China
under Grant 12301660 and the Science and Technology Commission of Shanghai
Municipality under Grant 23JC1400700.
\textit{Corresponding author: Tiandong Wang. The two authors contributed equally to this paper.}}%
\IEEEcompsocitemizethanks{\IEEEcompsocthanksitem Z. Li is with the 
Research Institute of Intelligent Complex Systems, Fudan University, 200433, Shanghai, China.
\IEEEcompsocthanksitem T. Wang is with the Shanghai Center for Mathematical Sciences, Fudan University, 200438, 
and Shanghai Academy of Artificial Intelligence for Science, Shanghai, China. \protect\\ 
E-mail: td\_wang@fudan.edu.cn.}
}

\markboth{IEEE Transactions on Network Science and Engineering,~Vol.~XX,~No.~X,~Month~YYYY}%
{Li and Wang: Spatial Dependence in Directed Preferential-Attachment Networks}

\maketitle

\begin{abstract}
Spatially embedded directed networks, such as airline networks, often exhibit
simultaneous high activity at nearby nodes. Preferential attachment (PA)
explains hub dominance. We extend it to spatial co-movement through a directed
PA model whose out- and in-node weights follow temporally persistent
Gaussian-process lognormal fields. Under sublinear PA, out-degree proportions
converge to explicit normalized powered weights, whereas self-loop exclusion
yields a coupled in-degree limit. We derive a strictly concave inverse that
recovers the in-weights from terminal degree proportions. For ordered network
histories, we develop a minorization--maximization (MM) weight estimator and
profile likelihood for the PA exponent; temporal pre-whitening and a spatial
quasi-likelihood estimate the
latent covariance. Simulations verify transmission of distance-decaying
dependence and show how random segment volume creates a distance-independent
common mode in raw degrees. An analysis of U.S. domestic flights (2015--2019)
separates network-wide volume variation from a short-range spatial component.
An observed-volume reconstruction reproduces the raw-degree common mode, and
the fitted field yields an exploratory co-exceedance transition scale of
roughly 150 km. A per-carrier analysis of European air traffic also reveals the
same decomposition.
\end{abstract}

\begin{IEEEkeywords}
Preferential attachment, spatial networks, high-activity dependence,
Gaussian processes, air transportation networks.
\end{IEEEkeywords}

\section{Introduction}

\IEEEPARstart{M}{any} real-world networks are spatially embedded and
directed, including transportation, communication, and power-grid
systems. In such systems, geography constrains connectivity while flows
have a strong directionality, and extreme states (e.g.\ severe
congestion, large-scale cancellations) usually
exhibit substantial spatial clustering and propagate along the network~\cite{bombelli2023}. For instance, the U.S. airport network is
dominated by a few major hubs, and critical for
mobility and economic activity, which has made it an important case
study in network science~\cite{guimera2005}. The same network is also the
dominant substrate for long-range disease spread, where importation risk
is governed by directed passenger flux and concentrated at hubs~\cite{colizza2006,brockmann2013}, so spatial co-movement of hub activity
is informative about correlated introduction pressure across nearby
airports. Preferential attachment (PA) explains hub dominance and heavy-tailed
degree distributions through cumulative advantage~\cite{barabasi1999,vanderhofstad2016}. Standard PA theory emphasizes degree
asymptotics. We study distance-decaying dependence among node activities.

Spatial extensions of PA incorporate spatial information typically through
distance-dependent attachment probabilities~\cite{flaxman2007,barthelemy2011}, and geometric inhomogeneous random
graphs (GIRG) provide an analytical framework for spatially embedded
networks~\cite{bringmann2019}. These models constrain connectivity by
distance and reproduce degree heterogeneity, but their node fitnesses
or weights are typically independent and identically distributed
(i.i.d.). They therefore do not directly encode a spatial covariance process
for temporally varying node fitnesses; dependence instead arises through the
graph geometry. Our target is the co-movement of unusually high activity at
spatially proximate nodes. Spatial-dependence analysis provides rank-based tools such as the
F-madogram and, in max-stable models, the extremal coefficient~\cite{ledford1996,schlather2003,cooley2006}, but these tools are usually
studied separately from network growth mechanisms. From an inference
perspective, parameter estimation for directed PA models typically
relies on maximum likelihood~\cite{wan2017,pham2016}, which becomes
costly and slowly convergent as the parameter dimension grows, and
existing weight estimators ignore both the correlation among weights
and edge directionality.

This paper embeds spatial high-activity dependence in PA and traces its
transmission from latent node attractiveness to observable network degrees.
Under sublinear PA, out-degree proportions converge to a powered and normalized
function of the latent out-weights. This map preserves latent spatial
dependence and amplifies hub heterogeneity. The main contributions are four-fold.
\begin{enumerate}
\item First, we embed a Gaussian-process driven
lognormal random field as latent node weights in a directed PA model,
such that the network exhibits hub-dominated degree heterogeneity and
distance-decaying finite-level spatial dependence.
\item Theoretically, we derive almost-sure limits for both degree-proportion
channels. Under sublinear PA, the out-degree limit is an explicit vector of
normalized powered weights, whereas self-loop exclusion makes the in-degree limit a
coupled variational solution. A strictly concave inverse recovers the in-weights
from terminal degree proportions. At the linear boundary, unique fitness
maxima produce winner-take-all behavior.
\item For diagnostics and inference, we use the F-madogram and the
madogram dependence coefficient to quantify network-level co-movement, and
we develop a comprehensive estimation procedure consisting of an MM weight step, a
profile-likelihood treatment of the PA exponent, first-order autoregressive
[AR(1)] pre-whitening of
log-ratio coordinates, and plug-in spatial covariance inference
on the resulting residual innovations.
\item For practical illustration, using data from the Bureau of Transportation
Statistics (BTS)~\cite{btsdata}, we model the U.S. domestic airline network and
show that degree co-movement decomposes into observed network-wide volume
variation (a distance-independent common mode)
and a short-range spatial component. A per-carrier analysis of European air
traffic from OpenSky~\cite{opensky2021,openskydata} reproduces this decomposition and an identifiable spatial
range of the same order, providing a geographic replication of the empirical
pattern.
\end{enumerate}

The rest of the paper is organized as follows.
Section~\ref{sec:related} positions the work relative to spatial network
models, spatial-dependence diagnostics, and PA inference. Section~\ref{sec:model}
presents the model and diagnostics. Section~\ref{sec:theory} states the
main theoretical properties. Section~\ref{sec:est} develops the estimation
framework. Section~\ref{sec:sim} reports simulations, Section~\ref{sec:real}
analyzes BTS aviation data, and Section~\ref{sec:concl} concludes. Proofs
are collected in the supplement.

\section{Related Work}
\label{sec:related}

\subsection{PA, Fitness, and Spatial Network Models}
Classical PA produces hubs through cumulative advantage~\cite{barabasi1999,vanderhofstad2016}. Intrinsic node fitness has been added
both to growing PA and to static connection models~\cite{bianconi2001,caldarelli2002}. In their canonical forms, fitness is
quenched and independently sampled, and the principal objects are marginal
degree distributions, scaling exponents, and condensation. Our setting
instead allocates repeated directed edge events among a fixed set of nodes
and lets their latent weights vary across time segments.

Spatial PA and related geometric models introduce location through a
distance-dependent attachment kernel or edge probability~\cite{flaxman2007,barthelemy2011}. GIRGs, for example, combine geometry with
heterogeneous node weights to reproduce spatial connectivity and degree
heterogeneity~\cite{bringmann2019}. In the canonical formulations, the node
weights are sampled independently and spatial dependence is induced by the
edge-formation kernel. Here distance does not enter the conditional
attachment probability directly. It enters through the covariance of the
time-varying latent weight field, so the model targets distance-decaying
co-movement of node activity rather than distance-decaying edge propensity.

\subsection{Fixed-Node Reinforcement}
Conditional on its latent weights, the source-selection process is closely
related to a finite-color nonlinear P\'olya urn. Past selections reinforce
future selection probabilities, with qualitatively different regimes as the
reinforcement exponent changes~\cite{chung2003}. In our weighted version,
sublinear feedback yields a deterministic sharing limit, whereas a unique
maximal weight yields monopoly at the linear boundary. This connection yields
the powered-weight out-degree limit. Self-loop exclusion couples target
selection to the realized source, producing a nonautonomous in-degree process.
We derive its coupled limit, strictly concave variational characterization,
and directed edge-frequency limit. The spatial Gaussian field induces
dependence across nodes and segments.

\subsection{Spatial-Dependence Diagnostics}
The F-madogram is a rank-based spatial-dependence diagnostic and, under
max-stability, is linked to the extremal coefficient~\cite{schlather2003,cooley2006,davison2012}. We use it as a finite-level
copula-dependence measure for latent weights and degree proportions, without a
max-stability assumption. The corresponding transmission result identifies
when the distance-decaying profile remains visible after PA reinforcement and
normalization.

\subsection{Inference for PA with Latent Heterogeneity}
Existing likelihood and joint attachment--fitness estimators primarily
address growing networks observed through ordered edge histories~\cite{wan2017,pham2016}. Coarse timestamps conceal that order. Poisson
edge-growth PA accommodates multiple events at each recorded time~\cite{wangresnick2023}. Our fixed-node model conditions on observed segment
volume and uses limit inversion when only terminal counts are available. Each
event contains two weighted choices, and self-loop exclusion couples the
target likelihood to the source sequence. For a fixed exponent, the
constrained log-weight problem is globally solvable by an MM algorithm;
profiling treats the exponent. Pre-whitened log-ratio coordinates then identify
the spatial covariance of temporally evolving node activity.

\section{Spatial Network Model}
\label{sec:model}

\subsection{Unified Segment-Level Specification}
We propose a hierarchical model combining a temporally persistent spatial
field of node weights, the observed number of edge events in each time segment,
and a conditional within-segment directed PA process. Let
$V=\{1,\ldots,N\}$ be a
fixed vertex set, with node $i$ located at $\bm{x}_i\in D\subset\R^d$,
and let $\ell=1,\ldots,m$ index observational segments.  Segment $\ell$
contains the observed number $T_\ell$ of directed edge events and starts
from zero within-segment degrees.

For $A\in\{\mathrm{out},\mathrm{in}\}$, define the
positive latent weight through a spatial Gaussian process (GP)~\cite{banerjee2014}:
\begin{equation}
\label{eq:spatial-composition}
\log w^{A}_{\ell i}=\mu^{A}(\bm{x}_i)+s^{A}_{\ell}(\bm{x}_i).
\end{equation}
The mean-zero spatial component follows
\begin{equation}
s^A_\ell(\bm{x}_i)=\phi_A s^A_{\ell-1}(\bm{x}_i)
+e^A_\ell(\bm{x}_i),
\qquad |\phi_A|<1,
\label{eq:temporal}
\end{equation}
where the innovation vectors are independent over $\ell$ and have
covariance
\[
\operatorname{Cov}\{e^A_\ell(\bm{x}_i),e^A_\ell(\bm{x}_j)\}
=(1-\phi_A^2)\gamma_A
\exp\{-\|\bm{x}_i-\bm{x}_j\|/\xi_A\}.
\]
Hence, the stationary marginal covariance of $s^A_\ell$ is
\begin{equation}
\label{eq:cov}
\Sigma_A=\Sigma(\bm{x};\gamma_A,\xi_A)
=\bigl\{\gamma_A\exp(-\|\bm{x}_i-\bm{x}_j\|/\xi_A)\bigr\}_{i,j=1}^{N},
\end{equation}
where $\gamma_A$ is the marginal variance and $\xi_A$ is the 
range parameter. Write
$\bm\mu^A=(\mu^A(\bm x_1),\ldots,\mu^A(\bm x_N))^\top$.
Equivalently, at a stationary segment,
\begin{align}
\bm w^{\mathrm{out}}_\ell&\sim\mathrm{Lognormal}
  (\bm\mu^{\mathrm{out}},\Sigma_{\mathrm{out}}),\label{eq:wo}\\
\bm w^{\mathrm{in}}_\ell&\sim\mathrm{Lognormal}
  (\bm\mu^{\mathrm{in}},\Sigma_{\mathrm{in}}),\label{eq:wi}
\end{align}
where the arguments denote the mean and covariance of the underlying
Gaussian log-vector. For a complete joint generator we take the out- and
in-innovation fields to be independent. The marginal results below do not
depend on that working assumption and the channels are fit separately;
their strong empirical association is therefore descriptive evidence for
a future bivariate-field extension, not a prediction of the present model.

Conditional on $(T_\ell,\bm w^{\mathrm{out}}_\ell,
\bm w^{\mathrm{in}}_\ell)$, let $G_\ell(t)$ contain the first $t$ edge
events of segment $\ell$. From $t-1$ to $t$, a directed edge $i\to j$,
$j\ne i$, is added with probability
\begin{align}
\label{eq:pa}
&\Prob_\ell\!\left(i\to j\mid G_\ell(t-1),T_\ell,\bm w_\ell\right)
\nonumber\\
&\quad=\frac{w^{\mathrm{out}}_{\ell i}
  (D^{\mathrm{out}}_{\ell i}(t-1)+1)^{\alpha}}
 {\sum_{k=1}^N w^{\mathrm{out}}_{\ell k}
  (D^{\mathrm{out}}_{\ell k}(t-1)+1)^{\alpha}}\nonumber\\
&\qquad\times
\frac{w^{\mathrm{in}}_{\ell j}
  (D^{\mathrm{in}}_{\ell j}(t-1)+1)^{\alpha}}
 {\sum_{k\ne i}w^{\mathrm{in}}_{\ell k}
  (D^{\mathrm{in}}_{\ell k}(t-1)+1)^{\alpha}},
\end{align}
and self-loops have probability zero. Here $\alpha\in[0,1]$ controls
within-segment degree reinforcement. Multiple edges are allowed, so the
result is a directed multigraph in which repeated $i\to j$ events
represent repeated interactions such as flights~\cite{diestel2010}.
Write $D^A_{\ell i}:=D^A_{\ell i}(T_\ell)$.
\begin{assumption}[Likelihood regularity]
\label{ass:likelihood}
For each network to which the weight likelihood is applied, $N\ge3$ and
every node has positive terminal out- and in-degree.
\end{assumption}

For theoretical and algorithmic statements concerning one generic
network, we suppress $\ell$ and write $T,w_i^A,D_i^A(t)$. The source
process is autonomous, which yields the powered-weight limit in
Section~\ref{sec:theory}. Self-loop exclusion makes the target process
nonautonomous. Its normalizing denominator depends on the current source,
coupling its drift to the source degree proportions. Theorem~\ref{thm:degree-tail}
characterizes the resulting in-degree limit by a coupled fixed point and an
equivalent strictly concave variational problem.

Define the observed segment volume and degree proportions by
\begin{equation}
\label{eq:volume}
T_\ell=\sum_{i=1}^{N}D^A_{\ell i},\qquad
Q^A_{\ell i}:=\frac{D^A_{\ell i}}{T_\ell},
\qquad A\in\{\mathrm{out},\mathrm{in}\}.
\end{equation}
Here $T_\ell$ is the segment's total network activity, and
$Q^A_{\ell i}$ is node $i$'s degree proportion.
For out-degrees and $0<\alpha<1$, the large-horizon result gives
$Q^{\mathrm{out}}_{\ell i}\approx
p_i(\bm w^{\mathrm{out}}_\ell,\alpha)$. Dividing by $T_\ell$ removes its
direct multiplicative scale, although finite-horizon sampling error and
any dependence between $T_\ell$ and the relative weights can remain.
We condition on the observed $T_\ell$ in the PA likelihood and do not
assign it a separate latent model. A stylized independence benchmark for
its contribution to raw-activity dependence is given in
Remark~\ref{cor:floor}.

The fixed vertex set separates PA reinforcement from scale-free degree
asymptotics. As edge volume grows, the $N$ degree proportions converge to a
finite vector. The usual growing-population degree distribution is
absent. PA therefore produces persistent hub heterogeneity without a
scale-free limit~\cite{barabasi1999,broido2019}. The lognormal layer supplies subexponential
latent heterogeneity and a Gaussian spatial copula. The rank-based diagnostics
below measure finite-level spatial co-movement.

\subsection{Quantifying Spatial Co-Movement}
To measure copula-level co-movement without imposing a common marginal scale,
we use the F-madogram~\cite{cooley2006,ribatet2013}:
\begin{equation}
\label{eq:fmado}
v_F(\bm{x}_1-\bm{x}_2)=\tfrac12\,\E\bigl[\,|F_{\bm{x}_1}\{Z(\bm{x}_1)\}
-F_{\bm{x}_2}\{Z(\bm{x}_2)\}|\,\bigr],
\end{equation}
with $F_{\bm{x}}$ the CDF of $Z(\bm{x})$. In practice, $F_{\bm{x}}$ is
replaced by the empirical CDF, giving
\begin{equation}
\label{eq:fmadohat}
\widehat{v}_F(\bm{x}_1-\bm{x}_2)=\frac{1}{2n(n+1)}\sum_{i=1}^{n}
\bigl|R_i(\bm{x}_1)-R_i(\bm{x}_2)\bigr|,
\end{equation}
where $R_i(\bm{x}_j)=\sum_{\ell=1}^{n}\mathbf{1}\{Z_\ell(\bm{x}_j)\le
Z_i(\bm{x}_j)\}$ is the rank of $Z_i(\bm{x}_j)$. The usual consistency
argument uses independent replicates, but for temporally dependent segments,
we use the statistic descriptively and preserve dependence in the block
bootstrap. For sites $\bm{x}_1,\dots,\bm{x}_k$ the
binned estimator is
\begin{equation}
\label{eq:binned}
\widehat{v}_{F,b}(h)=\frac{1}{|C_h|}\!\!\sum_{(i,j)\in C_h}
\!\!\widehat{v}_F(\bm{x}_i-\bm{x}_j),
\end{equation}
with
$C_h=\{(i,j):1\le i<j\le k,\ |\,\|\bm{x}_i-\bm{x}_j\|-h\,|<\delta\}$
for binning radius $\delta>0$.

We report the monotone re-expression
\begin{equation}
\label{eq:theta}
\theta_F=\frac{1+2v_F}{1-2v_F},
\end{equation}
which we call the \emph{madogram dependence coefficient}.  For a
max-stable process this transformation equals the pairwise extremal
coefficient~\cite{schlather2003,cooley2006}. For the Gaussian-copula
model used here, it is instead a finite-level rank-dependence index. Under
independence, we have $v_F=1/6$ and $\theta_F=2$, while complete positive
dependence gives $\theta_F=1$. Population values need not lie in
$[1,2]$ for an arbitrary negatively dependent copula, and finite-sample
estimates can also exceed $2$. 

In addition, we use
$h\mapsto\widehat\theta_F(h)$ as an operational summary. A rising
profile indicates stronger rank co-movement at short lags and weaker
rank co-movement at long lags~\cite{davison2012}. Because the F-madogram uses
all ranks, it is a finite-level copula diagnostic rather than a measure of
asymptotic tail dependence or a direct estimate of the latent covariance
range. Both rank diagnostics are applied to the lognormal
field and to the generated network in Section~\ref{sec:sim}.

\section{Theoretical Properties}
\label{sec:theory}

This section connects the latent weights to long-horizon degree and edge
frequencies. All limits are conditional on fixed positive weights, which
remain random across segments under the Gaussian-process layer. For
$0<\alpha<1$ the autonomous source process has an explicit limit, while
self-loop exclusion couples the target limit to the source degree proportions; at
$\alpha=1$ unique maximal weights produce winner-take-all limits. The
sublinear source limit is summarized by the transmission map
\begin{equation}
\label{eq:transmission-map}
\bm{w}\longmapsto
\bm{p}(\bm{w},\alpha),\qquad
p_i=\frac{w_i^{1/(1-\alpha)}}{\sum_{k=1}^{N}w_k^{1/(1-\alpha)}}.
\end{equation}
As $\alpha\uparrow1$ the exponent $r=1/(1-\alpha)$ concentrates the limiting
proportions on the largest weights, whereas $\alpha=0$ leaves attachment
fitness-driven. The limit identifies only the composite
$w_i^r/\sum_k w_k^r$, not $\alpha$ and $\bm w$ separately, so identifying
$\alpha$ needs the transient edge history (Remark~\ref{rem:alpha-ident}); the
unnormalized field $W^r$ retains the Gaussian copula and range of $W$ with
log-variance scaled by $r^2$, while the shared random denominator means that
the normalized vector need not retain that copula
(Proposition~\ref{prop:alpha-invariance}).

\begin{theorem}
\label{thm:degree-tail}
Consider the directed spatial preferential-attachment model in
\eqref{eq:pa}, and condition on fixed positive weights
$(\wout,\win)$. 

(i) Suppose $0<\alpha<1$ and define
\[
Y_i^A(T)=\frac{D_i^A(T)+1}{T+N},
\qquad A\in\{\mathrm{out},\mathrm{in}\}.
\]
Then as $T\to\infty$,
\[
Y^{\mathrm{out}}(T)\convas p^{\mathrm{out}},
\]
where
\[
p_i^{\mathrm{out}}=\frac{(\wout_i)^{1/(1-\alpha)}}
{\sum_{k=1}^{N}(\wout_k)^{1/(1-\alpha)}}.
\]
The in-degree proportions also converge almost surely, i.e.\
\[
Y^{\mathrm{in}}(T)\convas p^{\mathrm{in}},
\]
where $p^{\mathrm{in}}$ is the unique interior solution of
\begin{equation}
\label{eq:in-fixed-point}
p_j^{\mathrm{in}}=\win_j(p_j^{\mathrm{in}})^\alpha
\sum_{i\ne j}\frac{p_i^{\mathrm{out}}}
 {\sum_{k\ne i}\win_k(p_k^{\mathrm{in}})^\alpha},
\qquad j=1,\ldots,N.
\end{equation}
Equivalently,
\begin{equation}
\label{eq:in-variational}
p^{\mathrm{in}}=\argmax_{y\in\Delta_N}
\sum_{i=1}^{N}p_i^{\mathrm{out}}
\log\!\left(\sum_{k\ne i}\win_k y_k^\alpha\right),
\end{equation}
where $\Delta_N=\{y\in\mathbb R_+^N:\sum_i y_i=1\}$, and the objective
in \eqref{eq:in-variational} is strictly concave.
Hence, $D_i^A(T)/T\convas p_i^A$ for $A\in\{\text{out, in}\}$. For any threshold
$u$ satisfying $u\neq p_i^{\mathrm{out}}$ for all $i$,
\begin{align*}
\overline F_T^{\mathrm{out}}(u)
&=\frac{1}{N}\sum_{i=1}^{N}
\mathbf{1}\!\left\{\frac{\Dout_i(T)}{T}>u\right\}\\
&\longrightarrow
\frac{1}{N}\sum_{i=1}^{N}\mathbf{1}\{p_i^{\mathrm{out}}>u\}
\quad \text{a.s.}
\end{align*}
Under the GP-lognormal construction
\eqref{eq:spatial-composition}--\eqref{eq:wi}, the limiting out-degree
ranking and concentration are therefore governed by the powered latent field
$(\wout_i)^{1/(1-\alpha)}$, whose log-covariance is multiplied by
$(1-\alpha)^{-2}$.

(ii) For $\alpha=1$, assume that
\[
i_*:=\argmax_i\wout_i
\quad\text{and}\quad
j_*:=\argmax_{j\ne i_*}\win_j
\]
are unique. Then the linear PA setup leads to the winner-take-all
limits:
\[
Y^{\mathrm{out}}(T)\convas e_{i_*},\qquad
Y^{\mathrm{in}}(T)\convas e_{j_*}.
\]
Under the nondegenerate GP-lognormal specification, both maximizers are
unique almost surely.
\end{theorem}

The proofs (in the supplement) adapt classical stochastic-approximation and
urn arguments to the directed, self-loop-excluding setting. The linear-case
uniqueness condition is essential: if $M=\argmax_i\wout_i$ has more than one
node, $Y^{\mathrm{out}}(T)$ converges to a random vector on $M$ whose positive
coordinates are $\operatorname{Dirichlet}(1,\ldots,1)$, so $\alpha=1$ is a
genuine phase boundary rather than a continuous extension of the sublinear
formula. Self-loop exclusion keeps a deterministic in-degree limit but removes
the powered closed form: $p_j^{\mathrm{in}}\propto(\win_j)^{1/(1-\alpha)}$
holds exactly only without exclusion, and with loops excluded it is a
diffuse-network approximation (no fixed-$N$ bound claimed). At
$\alpha=0$, $p_i^{\mathrm{out}}=\wout_i/\sum_k\wout_k$ and the in-degree
limit has the explicit form
\[
p_j^{\mathrm{in}}=\win_j\sum_{i\ne j}
\frac{p_i^{\mathrm{out}}}{\sum_{k\ne i}\win_k}.
\]

A finer consequence is convergence of all directed edge frequencies.
\begin{theorem}
\label{thm:edge-measure}
Condition on the node locations $\{\bm{x}_i\}_{i=1}^{N}$ and on positive
latent weights $\{\wout_i,\win_i\}_{i=1}^{N}$. Let $0<\alpha<1$,
let $p^{\mathrm{out}}$ and $p^{\mathrm{in}}$ be the limits in
Theorem~\ref{thm:degree-tail}, and let $A_{ij}(T)$ be the number of
directed edges from $i$ to $j$ up to time $T$. Define
\[
K_{ij}=\mathbf{1}\{i\neq j\}
\frac{\wout_i(p_i^{\mathrm{out}})^\alpha}
{\sum_{a=1}^{N}\wout_a(p_a^{\mathrm{out}})^\alpha}
\frac{\win_j(p_j^{\mathrm{in}})^\alpha}
{\sum_{b\neq i}\win_b(p_b^{\mathrm{in}})^\alpha}.
\]
Then $A_{ij}(T)/T\convas K_{ij}$ for every ordered pair $(i,j)$ with
$i\neq j$. Equivalently, the empirical directed edge measure
\[
\Gamma_T=\sum_{i\neq j}\frac{A_{ij}(T)}{T}\delta_{(\bm{x}_i,\bm{x}_j)}
\]
converges almost surely, in the weak sense of finite measures, to
$\Gamma=\sum_{i\neq j}K_{ij}\delta_{(\bm{x}_i,\bm{x}_j)}$.

For $\alpha=1$, under the unique-maximizer conditions of
Theorem~\ref{thm:degree-tail},
\[
\Gamma_T\convas
\delta_{(\bm{x}_{i_*},\bm{x}_{j_*})}.
\]
\end{theorem}

\begin{remark}{\rm
The limit in Theorem~\ref{thm:edge-measure} is a finite atomic
directed-multigraph limit: the vertex set is fixed and only the edge count
grows, unlike the GIRG regime where $N\to\infty$. It does not establish an
$N\to\infty$ limit, but formally, if the empirical location measure converges
to a reference measure $\mu$ on $D$, the weight fields admit corresponding
continuum versions, and exclusion of the diagonal is asymptotically
negligible, the associated rank-one edge-intensity density would be
\[
\begin{aligned}
\mathcal K(x,y)&=\pi_{\mathrm{out}}(x)\pi_{\mathrm{in}}(y),\\
\pi_A(x)&=\frac{\{W^A(x)\}^{1/(1-\alpha)}}
{\int_D\{W^A(u)\}^{1/(1-\alpha)}\,d\mu(u)},
\quad A\in\{\mathrm{out},\mathrm{in}\}.
\end{aligned}
\]
Establishing this continuum limit requires a separate joint $N,T$ asymptotic
theory. Even formally, $\mathcal K$ is not a standard GIRG kernel, as dependence
on distance enters through correlated random fields rather than direct decay
in $\|x-y\|$. A GIRG-like extension would insert a
distance kernel $\varphi(\|\bm{x}_i-\bm{x}_j\|)$ into the attachment
probability.}
\end{remark}

The next result transfers the coordinate-wise degree limit to the population
F-madogram.
\begin{theorem}
\label{thm:transfer-degrees}
Condition on fixed node locations $\{\bm{x}_i\}_{i=1}^{N}$ and let
$\{W_i^{\mathrm{out}}\}_{i=1}^{N}$ be positive latent out-weights generated
by the GP-lognormal field in \eqref{eq:spatial-composition}--\eqref{eq:wi}. Suppose
$0<\alpha<1$, define $r=1/(1-\alpha)$, and let
\[
X_i^{(T)}=\frac{D_i^{\mathrm{out}}(T)}{T},\qquad
P_i=\frac{(W_i^{\mathrm{out}})^r}
{\sum_{k=1}^{N}(W_k^{\mathrm{out}})^r}.
\]
Then, for every node $i$, $X_i^{(T)}\convas P_i$, as
$T\to\infty$. Thus, for any pair of sites $(\bm{x}_i,\bm{x}_j)$
whose limiting marginal distributions are continuous, the pairwise
F-madogram of the out-degree-proportion field converges to that of the
transformed weight field,
\[
v_{F,T}^{\mathrm{deg}}(\bm{x}_i,\bm{x}_j)\to
v_F^P(\bm{x}_i,\bm{x}_j),
\]
and the corresponding madogram coefficient converges under the mapping
$\theta_F=(1+2v_F)/(1-2v_F)$.
\end{theorem}

At segment $\ell$, Theorem~\ref{thm:transfer-degrees} is applied conditionally
at $T=T_\ell$, so $D_{\ell i}^{\mathrm{out}}(T_\ell)/T_\ell$ is the
finite-horizon counterpart of $p_i(\bm{w}_\ell^{\mathrm{out}},\alpha)$; along
any almost-surely diverging sequence of random horizons the limit follows
pathwise, covering cumulative Poisson edge counts~\cite{wangresnick2023}, and
conditioning on the observed $T_\ell$ leaves its distribution unspecified. The
limit $P$ has a shared denominator that generally changes the copula, so its
F-madogram has no Gaussian closed form; the next proposition instead gives an
exact benchmark for the \emph{unnormalized} powered field.

\begin{proposition}
\label{prop:latent-fmado}
Let $Z(\bm{x})=g\{G(\bm{x})\}$ be any strictly increasing transform of a
Gaussian field with arbitrary deterministic mean, constant marginal variance,
and correlation
$\rho(h)=\exp(-h/\xi)$. Then the pairwise F-madogram of $Z$ depends only
on $\rho(h)$ and is
\begin{equation}
\label{eq:vF-latent}
v_F(h)=\frac14-\frac{1}{2\pi}
\arcsin\!\left(\frac{1+\rho(h)}{2}\right),
\end{equation}
with corresponding madogram-coefficient benchmark
\begin{equation}
\label{eq:theta-latent}
\theta_{F,\mathrm{lat}}(h)
=\frac{1+2v_F(h)}{1-2v_F(h)}.
\end{equation}
Moreover, this benchmark is invariant to the monotone link $g$, the
marginal variance, and the PA exponent when it is applied to
$Z=(W^{\mathrm{out}})^{1/(1-\alpha)}$. 
\end{proposition}

\begin{remark}
\label{cor:floor}
{\rm Let
$
V_\ell=\log T_\ell-\E(\log T_\ell)
$
be the centered log-volume of segment $\ell$. To isolate the effect of
$T_\ell$, consider the approximation
\begin{equation}
\label{eq:raw-degree-benchmark}
\log A_{\ell i}
=
\mu_i+V_\ell+r s_\ell(\bm{x}_i),
\qquad r=\frac{1}{1-\alpha},
\end{equation}
where $A_{\ell i}>0$ represents the raw out-degree at node $i$,
$V_\ell\sim\Nrm(0,\sigma_T^2)$, and $V_\ell$ is independent of the
Gaussian spatial field $s_\ell$. Assume
\[
\operatorname{Cov}\{s_\ell(\bm{x}_i),s_\ell(\bm{x}_j)\}
=
\gamma e^{-h/\xi},
\qquad
h=\|\bm{x}_i-\bm{x}_j\|.
\]
Then the
large-horizon approximation from Theorem~\ref{thm:degree-tail} is
\[
\log D_{\ell i}^{\mathrm{out}}
\approx
\log T_\ell
+r\log w_{\ell i}^{\mathrm{out}}
-\log\sum_{k=1}^{N}
   (w_{\ell k}^{\mathrm{out}})^r.
\]
Therefore, \eqref{eq:raw-degree-benchmark} retains the observed volume term and
the PA amplification factor $r$, but omits the final random normalizing
term.

Under \eqref{eq:raw-degree-benchmark}, the correlation between nodes at
distance $h$ is
\begin{equation}
\label{eq:rho-tot}
\rho_{\mathrm{tot}}(h)
=
\frac{\sigma_T^2+r^2\gamma e^{-h/\xi}}
     {\sigma_T^2+r^2\gamma}.
\end{equation}
Substituting $\rho_{\mathrm{tot}}(h)$ into
\eqref{eq:vF-latent} gives the corresponding Gaussian-copula
F-madogram. In particular,
\[
\rho_{\mathrm{tot}}(h)\longrightarrow
\rho_\infty
=
\frac{\sigma_T^2}{\sigma_T^2+r^2\gamma}
\qquad\text{as }h\to\infty.
\]
When $\sigma_T^2>0$, this positive limiting correlation makes the
madogram dependence coefficient approach a value below its independence
reference value of $2$. We call this distance-independent component the
\emph{common mode}, i.e.\ the variation in $T_\ell$ affects every node in the same
segment, regardless of distance.

Dividing the raw degrees by $T_\ell$ removes $V_\ell$ exactly. It does
not, however, remove finite-horizon error, the random normalizing term,
or possible dependence between $T_\ell$ and the node weights.}
\end{remark}

Remark~\ref{cor:floor} explains why raw-degree dependence may remain
positive at large distances. It is a diagnostic approximation, not an
exact formula for the dependence generated by the PA model.

We now explain the effect of $\alpha$ on the inferred spatial range.
\begin{proposition}
\label{prop:alpha-invariance}
Fix an interior limiting out-degree-proportion vector
$p=(p_1,\ldots,p_N)$. Suppose two sets of weights reproduce $p$ through
\eqref{eq:transmission-map} at
$\alpha,\alpha'\in(0,1)$. Then
\begin{equation}
\label{eq:alpha-weight-rescaling}
\log w_i(\alpha')-\log w_j(\alpha')
=
\frac{1-\alpha'}{1-\alpha}
\left\{
\log w_i(\alpha)-\log w_j(\alpha)
\right\}
\end{equation}
for every $i$ and $j$. Equivalently,
\[
\log w_i(\alpha')
=
\frac{1-\alpha'}{1-\alpha}\log w_i(\alpha)+c
\]
for a constant $c$ common to all nodes.

Hence, changing $\alpha$ in this reconstruction only rescales the
log-weight difference. Under the Gaussian spatial model, this rescaling
leaves the correlation function and the range parameter $\xi$ unchanged,
while the marginal variance changes according to
\[
\gamma(\alpha')
=
\left(\frac{1-\alpha'}{1-\alpha}\right)^2
\gamma(\alpha).
\]
\end{proposition}

Proposition~\ref{prop:alpha-invariance} concerns weights reconstructed from
the same limiting out-degree-proportion vector; it does not imply that
degree-proportion
dependence is unchanged when $\alpha$ varies with the latent weight
distribution held fixed, and need not hold exactly for finite-horizon MM
estimates.

\section{Parameter Estimation}
\label{sec:est}

The parameters to be estimated are the mean functions
$\bm{\mu}^{\mathrm{out}}(\bm{x})$, $\bm{\mu}^{\mathrm{in}}(\bm{x})$ and
the temporal-spatial parameters
$\Theta^{A}=(\phi_A,\gamma_A,\xi_A)$,
$A\in\{\mathrm{out},\mathrm{in}\}$. The
likelihood depends on the unobserved node weights and
couples them through nonlinear normalizers, so joint maximization over
all parameters is intractable. 

Instead, we adopt a
three-step strategy: (1) derive and separate the likelihood and, with
the PA exponent $\alpha$ fixed, estimate the weights by an MM scheme
(if $\alpha$ is unknown, the weights and $\alpha$ are updated
alternately by profile likelihood); (2) transform normalized log-weights
to identifiable log-ratio coordinates and fit their AR(1) temporal
component; (3) apply a plug-in spatial Gaussian quasi-likelihood to the
pre-whitened residuals. 

Step~1 uses the ordered-history MM estimator when a reliable edge sequence
is available. For terminal segment counts, as in the airline application,
limit inversion replaces sequence reconstruction. Theorem~\ref{thm:degree-tail}
gives explicit out-weights, and the exact no-self-loop inverse below gives the
in-weights.

This procedure estimates the spatiotemporal weight model
\eqref{eq:spatial-composition}--\eqref{eq:volume}. The $\ell_1$ normalization in
Step~1 (Algorithm~\ref{alg:weights}) removes the unidentified common weight
scale and the direct multiplicative scale of $T_\ell$. Steps~2 and~3 then
estimate residual innovations of $s$ after temporal pre-whitening. Since
the AR coefficient rescales variance but not spatial correlation, their
range $\xi$ is also the stationary range of $s$. Segment volumes are not
assigned additional latent parameters, but are observed totals used to
separate raw activity from relative activity. The random-volume experiment in
Fig.~\ref{fig:volume-sim} illustrates how this observed common scale can
create a raw-degree common mode and how normalization removes its direct effect.

\subsection{Terminal-Count Inversion}

Let $a\in\operatorname{int}(\Delta_N)$ and
$p\in\operatorname{int}(\Delta_N)$ denote limiting out- and in-degree
proportions. The out-channel inversion is
$w_i^{\mathrm{out}}\propto a_i^{1-\alpha}$. For the in channel, define
the common-scale-invariant log-mass criterion
\begin{equation}
\label{eq:in-inverse-objective}
\mathcal J_{a,p}(v)
=\sum_{j=1}^{N}p_jv_j
-\sum_{i=1}^{N}a_i\log\!\left(\sum_{k\ne i}e^{v_k}\right),
\qquad \onev^\top v=0.
\end{equation}

\begin{proposition}[Exact inversion of the no-self-loop limit]
\label{prop:in-inversion}
Suppose $N\ge3$, $0\le\alpha<1$, and
$0<p_j<1-a_j$ for every $j$. Then
\eqref{eq:in-inverse-objective} has a unique maximizer $\widehat v$.
The normalized weights
\begin{equation}
\label{eq:in-inverse-weights}
\widehat w_j^{\mathrm{in}}
=N\frac{e^{\widehat v_j}p_j^{-\alpha}}
{\sum_{k=1}^{N}e^{\widehat v_k}p_k^{-\alpha}}
\end{equation}
reproduce $p$ exactly through the fixed point
\eqref{eq:in-fixed-point} with source degree-proportion vector $a$. Conversely, every
positive in-weight vector that reproduces $p$ is proportional to the
vector in \eqref{eq:in-inverse-weights}.
\end{proposition}

The feasibility condition $p_j<1-a_j$ expresses the self-loop constraint:
node $j$ can receive an edge only when it is not the source. The objective
is strictly concave in the identifiable log-mass coordinates, so a damped
Newton method gives the unique solution. This is an exact inversion of the
\emph{asymptotic} target limit; terminal counts at finite horizon still
carry sampling error. A proof is given in the supplement.

\subsection{Likelihood and Its Properties}
The following formulas describe one generic network. When they are
applied to segment $\ell$, we set $T=T_\ell$,
$w_i^A=w_{\ell i}^A$, and $D_i^A(t)=D_{\ell i}^A(t)$, and suppress the
segment index to keep the likelihood readable.
Let $e_t=(I^{(1)}_t,I^{(2)}_t)$ be the edge formed from $G(t-1)$ to
$G(t)$. Given $G(0)$ and $\{e_t\}_{t=1}^{T}$, treating the latent
weights as variables with $0\le\alpha\le1$, the likelihood is
\begin{align}
&L\bigl(\wout,\win,\alpha\mid G(0),(e_t)_{t=1}^{T}\bigr)\nonumber\\
&=\prod_{t=1}^{T}\frac{\wout_{I^{(1)}_t}(\Dout_{I^{(1)}_t}(t-1)+1)^{\alpha}}
{\sum_{k\in[N]}\wout_k(\Dout_k(t-1)+1)^{\alpha}}\nonumber\\
&\qquad\times\frac{\win_{I^{(2)}_t}(\Din_{I^{(2)}_t}(t-1)+1)^{\alpha}}
{\sum_{k\neq I^{(1)}_t}\win_k(\Din_k(t-1)+1)^{\alpha}}.
\end{align}
Up to terms independent of the respective weights, the log-likelihood
separates as $\ell(\wout)+\ell(\win)$ with
\begin{align}
\ell(\wout)&=\sum_{i\in[N]}\Dout_i(T)\log\wout_i\nonumber\\
&\quad-\sum_{t=1}^{T}\log\!\Bigl(\sum_{k\in[N]}\wout_k(\Dout_k(t-1)+1)^{\alpha}\Bigr),
\label{eq:llo}\\
\ell(\win)&=\sum_{i\in[N]}\Din_i(T)\log\win_i\nonumber\\
&\quad-\sum_{t=1}^{T}\log\!\Bigl(\sum_{k\neq I^{(1)}_t}\win_k(\Din_k(t-1)+1)^{\alpha}\Bigr).
\label{eq:lli}
\end{align}
The omitted numerator-degree term is constant in the weights but not in
$\alpha$. Therefore, profiling the PA exponent must use the full
log-likelihood:
\begin{align}
\ell_{\mathrm{full}}(\wout,\win,\alpha)
&=\ell(\wout)+\ell(\win)+\alpha S_D,\label{eq:llfull}\\
\text{with }S_D&=\sum_{t=1}^{T}\left[
\log\{\Dout_{I^{(1)}_t}(t-1)+1\}\right]\nonumber\\
&\quad+\sum_{t=1}^{T}
\log\{\Din_{I^{(2)}_t}(t-1)+1\}.\nonumber
\end{align}

We now show that, after reparameterizing the node weights on the log scale, the constrained likelihood has a unique global maximizer.
\begin{theorem}
\label{thm:opt}
Under Assumption~\ref{ass:likelihood}, for $\alpha$ fixed and parameter
space $\mathcal{W}=(0,\infty)^N$, apply the log-transformation
$z=\log w$. Then \eqref{eq:llo} and \eqref{eq:lli} are concave in $z$,
and each admits a global optimizer after scale normalization. 

(i) Any
stationary point of the transformed objective, i.e.\ any
$w$ with $\partial\ell(\wout)/\partial\wout_i=0$ and
$\partial\ell(\win)/\partial\win_i=0$ for all $i$, is a global
optimizer. 

(ii) The optimizer is not unique. If $(\wout_0,\win_0)$ is
a solution, then every global optimizer has the form
$\wout_*=c_1\wout_0$, $\win_*=c_2\win_0$ with $c_1,c_2\in\R_+$. The
$\ell_1$ normalization $\|\bm{w}\|_1=N$ selects the unique
representative on each ray.
\end{theorem}

We emphasize that concavity and global optimality hold for the
\emph{transformed} objective under the log-transformation, not for the
original product likelihood, and that uniqueness holds only under the
$\ell_1$ normalization. Theorem~\ref{thm:opt} is proved in the
supplement via the concavity of the (generalized) log-sum-exp function.

\subsection{MM Algorithm for the Latent Weights (Known $\alpha$)}
Note that the likelihood equations have no closed form. By the
minorization--maximization (MM) principle~\cite{ortega2000} we replace
the hard problem with a sequence of tractable surrogates. The only
nonlinear term in \eqref{eq:llo} is
$-\sum_{t}\log\sum_{k}\wout_k(\Dout_k(t-1)+1)^{\alpha}$. Since
$-\log(\cdot)$ is convex, Jensen's inequality with weights
$p_{t,k}\propto\wout_k(s)(\Dout_k(t-1)+1)^{\alpha}$ yields a surrogate
that is tangent at the current iterate and, after the elementary
inequality $\log x\le x-1$, leads to the multiplicative updates
\begin{align}
\wout_i(s{+}1)&=\frac{\Dout_i(T)}
{\displaystyle\sum_{t=1}^{T}\frac{(\Dout_i(t-1)+1)^{\alpha}}
{\sum_{k\in[N]}\wout_k(s)(\Dout_k(t-1)+1)^{\alpha}}},
\label{eq:mmo}\\
\win_i(s{+}1)&=\frac{\Din_i(T)}
{\displaystyle\sum_{\{t:\,i\neq I^{(1)}_t\}}\frac{(\Din_i(t-1)+1)^{\alpha}}
{\sum_{k\neq I^{(1)}_t}\win_k(s)(\Din_k(t-1)+1)^{\alpha}}}.
\label{eq:mmi}
\end{align}
After each update the weights are renormalized to $\|\bm{w}\|_1=N$.

\begin{proposition}[MM minorizer]
\label{prop:mm}
At iteration $s$, with $a_{i,t}:=(\Dout_i(t-1)+1)^{\alpha}$, the
surrogate
\begin{align*}
g_1\bigl(\wout\mid\wout(s)\bigr)
&=\sum_{i\in[N]}\Dout_i(T)\log\wout_i\\
&\;-\sum_{t=1}^{T}\log\Bigl(\sum_{i\in[N]}\wout_i(s)\,a_{i,t}\Bigr)\\
&\;-\sum_{t=1}^{T}\frac{\sum_{i}\wout_i\,a_{i,t}}
{\sum_{i}\wout_i(s)\,a_{i,t}}+T
\end{align*}
minorizes $\ell(\wout)$ and is tangent at $\wout(s)$. An analogous
$g_2(\win\mid\win(s))$ holds with the summation restricted to
$i\neq I^{(1)}_t$.
\end{proposition}

The tangency, minorization, and convergence arguments are given in
the supplement. We now give the statement for the convergence of the 
MM algorithm.

\begin{theorem}
\label{thm:mm}
The MM updates
\eqref{eq:mmo}--\eqref{eq:mmi} generate a non-decreasing
and convergent sequence of log-likelihood values, and every limit
point is a stationary point of the original log-likelihood. Under the
$\ell_1$ normalization, monotonicity confines the iterates to an upper
level set that is compact in the positive simplex, so by
Theorem~\ref{thm:opt} the limit is the unique normalized optimizer.
\end{theorem}

In particular, the in-degree update procedure excludes the source $I^{(1)}_t$ from the denominator. This
is implemented with the time-specific attachment
\begin{equation}
\label{eq:dtilde}
\widetilde a^{\mathrm{in}}_{i,t}
=\mathbf{1}\{i\neq I^{(1)}_t\}
\bigl(\Din_i(t-1)+1\bigr)^{\alpha},
\end{equation}
which zeros only the current source mass and leaves the in-degree history
unchanged. Algorithm~\ref{alg:weights} summarizes the procedure.

\begin{algorithm}[h]
\caption{Latent-weight estimation (known $\alpha$)}
\label{alg:weights}
\begin{algorithmic}[1]
\STATE \textbf{Input:} $\{G(t)\}_{t=1}^{T}$, $V=\{1,\dots,N\}$, edges
$e_t=(I^{(1)}_t,I^{(2)}_t)$, $\alpha$, MM iterations $S$.
\STATE Initialize $\wout(0),\win(0)$.
\FOR{$t=1$ to $T$}
  \STATE $\Dout_i(t)\!\leftarrow\!\Dout_i(t-1)+\mathbf{1}\{i=I^{(1)}_t\}$
  \STATE update $\Din_i(t)$ and form the mask \eqref{eq:dtilde}
\ENDFOR
\FOR{$s=1$ to $S$}
  \STATE update $\wout$ by \eqref{eq:mmo}
  \STATE update $\win$ by \eqref{eq:mmi} using the masked masses
\ENDFOR
\STATE \textbf{return} $\widehat\wout=N\,\wout(S)/\|\wout(S)\|_1$,
$\widehat\win=N\,\win(S)/\|\win(S)\|_1$.
\end{algorithmic}
\end{algorithm}

\subsection{Profile Likelihood for the PA Exponent (Unknown $\alpha$)}
In practice, $\alpha$ is often unknown. Since the latent weights are
unobserved and $\alpha$ enters the log-likelihood in a nonlinear way together
with the weights, jointly maximizing $(\wout,\win,\alpha)$ is
difficult. Moreover, the concavity that holds for the weights alone is
lost once $\alpha$ is also free. We therefore use a profile-likelihood
approach. For fixed $\alpha$, let
\begin{equation}
\bigl(\widehat\wout(\alpha),\widehat\win(\alpha)\bigr)
=\argmax_{\wout,\win}\;\ell_{\mathrm{full}}
  \bigl(\wout,\win,\alpha\bigr)
\end{equation}
be the MM estimate from Algorithm~\ref{alg:weights}, and define the
profile log-likelihood
$\ell_p(\alpha):=\ell_{\mathrm{full}}\bigl(\widehat\wout(\alpha),
\widehat\win(\alpha),\alpha\bigr)$. Since $\ell_p$ has no closed-form
maximizer, we combine a grid search with local continuous
optimization, alternating weight and $\alpha$ updates until
convergence (Algorithm~\ref{alg:alpha}).

\begin{algorithm}[h]
\caption{Profile-likelihood estimation of $\alpha$}
\label{alg:alpha}
\begin{algorithmic}[1]
\STATE \textbf{Input:} edges $\{e_t\}_{t=1}^{T}$, grid
$\{\alpha_1,\dots,\alpha_K\}\subset[\alpha_{\min},\alpha_{\max}]$, MM
iterations $S$.
\STATE $\ell_{\max}\leftarrow-\infty$
\FOR{$k=1$ to $K$}
  \STATE run Algorithm~\ref{alg:weights} at $\alpha_k$ to get
  $\wout(\alpha_k),\win(\alpha_k)$
  \STATE $\ell_p(\alpha_k)\leftarrow
  \ell_{\mathrm{full}}(\wout(\alpha_k),\win(\alpha_k),\alpha_k)$
  \IF{$\ell_p(\alpha_k)>\ell_{\max}$}
    \STATE $\ell_{\max}\leftarrow\ell_p(\alpha_k)$;\;
    $\widehat\alpha\leftarrow\alpha_k$
  \ENDIF
\ENDFOR
\STATE refine: $\widehat\alpha\leftarrow\argmax_{\alpha\in
[\underline\alpha,\overline\alpha]}
\ell_{\mathrm{full}}(\widehat\wout(\alpha),\widehat\win(\alpha),\alpha)$ by Brent's
method in a neighborhood of $\widehat\alpha$
\STATE \textbf{return} $\widehat\alpha$,
$\widehat\wout(\widehat\alpha)$, $\widehat\win(\widehat\alpha)$.
\end{algorithmic}
\end{algorithm}

\begin{remark}[Weak stationary identification of $\alpha$]
\label{rem:alpha-ident}
{\rm
Long-run degree proportions do not identify $\alpha$. Indeed, if
$p=(p_1,\ldots,p_N)$ is the limiting out-degree vector, then for any
$\alpha'\in(0,1)$ the weights
\[
w_i(\alpha')=c\,p_i^{\,1-\alpha'},\qquad c>0,
\]
produce exactly the same limit through \eqref{eq:transmission-map}.
Therefore, $\alpha$ cannot be separated from the node weights using
long-run degree proportions alone.

Information about $\alpha$ comes from the edge sequence before the degrees
stabilize, because the attachment probabilities contain
$(D_i(t-1)+1)^\alpha$ at every step. Hence, profile-likelihood
estimation of $\alpha$ requires reliably ordered edge events and may remain
weak at moderate horizons. Later in Section~\ref{sec:real}, the airline data contain terminal daily counts
but no reliable within-day ordering, so we fix $\alpha$ and assess
sensitivity across plausible values. Under the large-horizon reconstruction,
changing $\alpha$ rescales the out-weight variation and leaves its fitted
spatial range unchanged, as stated in
Proposition~\ref{prop:alpha-invariance}. The exact in-weight inverse in
Proposition~\ref{prop:in-inversion} need not be a common rescaling, so its
range sensitivity is evaluated numerically.}
\end{remark}

\subsection{Temporal Pre-whitening and Spatial Quasi-Likelihood}
The MM weights are normalized within each segment. Directly assigning a
full Gaussian distribution to their logarithms would ignore the random
normalizing shift. Let $H\in\R^{N\times q}$, $q=N-1$, satisfy
$H^\top H=I_q$ and $H^\top\onev=0$, and define the identifiable
orthonormal log-ratio coordinates
\begin{equation}
C_\ell=H^\top Y_\ell,\qquad
Y_\ell=\log\widehat{\bm w}_\ell.
\label{eq:contrasts}
\end{equation}
Since $H^\top\onev=0$, each entry of $C_\ell$ is a zero-sum linear
combination of the log-weights, so the unknown common log-weight
shift is annihilated by $H^\top$ and only
the identifiable relative weights remain. Let
$\bar C=m^{-1}\sum_{\ell=1}^{m}C_\ell$ and
$\widetilde C_\ell=C_\ell-\bar C$. The empirical pooled AR estimate is
\begin{equation}
\widehat\phi=
\frac{\sum_{\ell=2}^{m}\widetilde C_{\ell-1}^{\top}
\widetilde C_\ell}
{\sum_{\ell=2}^{m}\|\widetilde C_{\ell-1}\|^2},
\label{eq:phi-hat}
\end{equation}
and temporal pre-whitening gives
\begin{equation}
E_\ell=\widetilde C_\ell-\widehat\phi\widetilde C_{\ell-1},
\qquad \ell=2,\ldots,m.
\label{eq:innovations}
\end{equation}
After subtracting the true mean and filtering with the true
$\phi$, the latent innovations $\varepsilon_\ell=H^\top e_\ell$ are i.i.d.
with covariance
\begin{equation}
\operatorname{Cov}(\varepsilon_\ell)=\gamma_e R_C(\xi),\qquad
R_C(\xi)=H^\top\Sigma_0(\xi)H,
\label{eq:innovation-cov}
\end{equation}
where $\gamma_e=(1-\phi^2)\gamma$ and
$\Sigma_0(\xi)=\{\exp(-\|\bm{x}_i-\bm{x}_j\|/\xi)\}$. Thus, the spatial
range is unchanged by pre-whitening, while stationary and innovation
variances must be distinguished.

For $n_e=m-1$ innovations, define
\[
Q(\xi)=\sum_{\ell=2}^{m}E_\ell^\top R_C(\xi)^{-1}E_\ell,
\qquad
\widehat\gamma_e(\xi)=\frac{Q(\xi)}{n_e q}.
\]
We estimate $\xi$ by minimizing the plug-in Gaussian profile objective
\begin{equation}
n_e\log|R_C(\xi)|+n_e q\log\widehat\gamma_e(\xi),
\label{eq:innovation-profile}
\end{equation}
and report the stationary variance
$\widehat\gamma=\widehat\gamma_e/(1-\widehat\phi^2)$. Since $\bar C$,
$\phi$, and the latent weights are estimated, the fitted $E_\ell$ are
residuals,
not exactly i.i.d. innovations, and \eqref{eq:innovation-profile} is therefore a
conditional quasi-likelihood. 

We use residual autocorrelation functions (ACFs) to assess its adequacy. Full
uncertainty requires a parametric or moving-block bootstrap that repeats
weight construction, AR fitting, and spatial fitting in every replicate;
intervals that omit any of these stages are labeled conditional. The
supplement gives the likelihood details and empirical diagnostics.

\section{Simulation Studies}
\label{sec:sim}

Two experiments test the mechanisms used in the empirical analysis. The first
tracks the transmission of latent spatial dependence to degree proportions;
the second isolates the common mode induced by random segment volume.
Supplementary experiments compare independent-weight baselines, examine both
degree directions, and assess weight recovery and identification of the PA
exponent. All experiments use fresh simulated fields and networks. 

\subsection{Settings and Network Generation}
Locations are drawn uniformly over $[0,10]^2\subset\R^2$ with the
Euclidean metric. Unless otherwise stated, weights are generated from
the GP-lognormal model \eqref{eq:wo}--\eqref{eq:wi} with
$\alpha_0=0.5$, $\gamma=1$, and $\xi=2$. Networks are
generated by sampling each endpoint from a categorical distribution
proportional to $w^{A}\odot(D^{A}+1)^{\alpha}$. The basic generation
procedure and a Fenwick-tree implementation with $O(T\log N)$ sampling for
larger runs are given in the supplement.

\subsection{Transmission of Latent Dependence to Degrees}
\label{sec:transmission-sim}
We first consider Theorem~\ref{thm:transfer-degrees}. For $N=80$ fixed
sites and $30$ independent latent fields, we compute the limiting
degree-proportion map \eqref{eq:transmission-map} and simulate PA networks at
$T\in\{1000,5000,20000\}$ from the same weights. Fig.~\ref{fig:transmission}
shows that the degree-level madogram-coefficient curves have the same
distance-decaying shape as the PA limit and move toward it as $T$
increases. The convergence occurs before any weights or spatial parameters are
estimated and therefore reflects the PA dynamics directly.

\begin{figure}[t]\centering
\includegraphicssafe[width=\columnwidth]{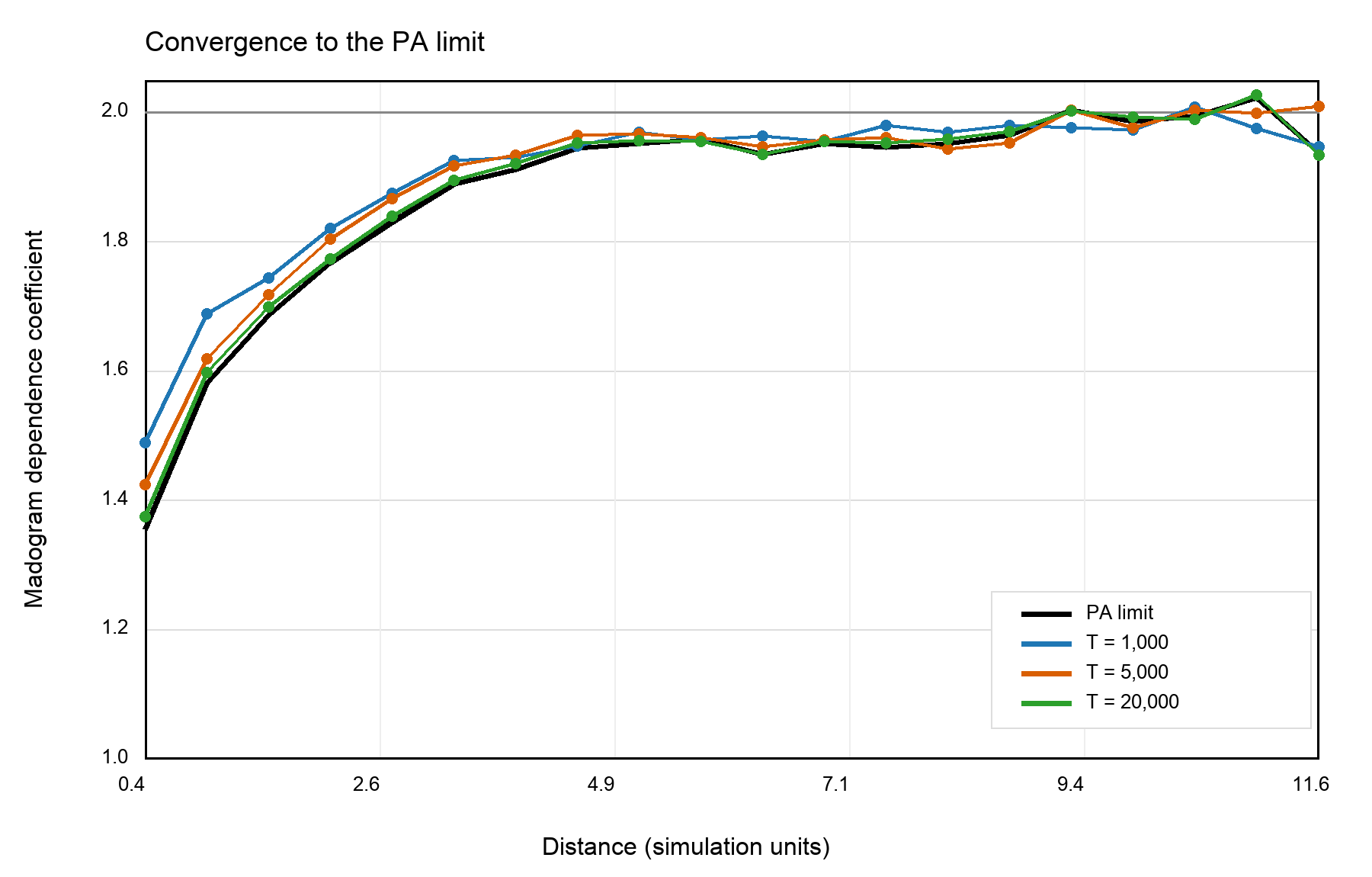}
\caption{Transmission experiment: binned madogram dependence coefficient of the PA
limit and of simulated degree proportions at increasing edge horizons.
The degree curves approach the powered-and-normalized latent-weight curve,
illustrating Theorem~\ref{thm:transfer-degrees}. The grey line marks the
independence reference $\theta_F=2$.}
\label{fig:transmission}
\end{figure}

The supplement reports the corresponding controls. Independent heavy-tailed
weights and spatially varying independent marginals produce nearly flat
profiles, whereas GP-correlated weights produce distance decay. A full
directed simulation also shows the same qualitative transmission in the
coupled in-degree process.

\subsection{Random Segment Volume}
We next test the observed-volume decomposition directly.
On the same $[0,10]^2$ domain, we generate $180$ segment replicates with
$N=80$, $\alpha=0.5$, $\gamma=0.03$, and $\xi=2$. For each segment,
we draw the spatial log-weight field independently of a random log-volume
$V_\ell=0.18Z_\ell+0.08\varepsilon_\ell$, where
$Z_\ell,\varepsilon_\ell\stackrel{\mathrm{i.i.d.}}{\sim}\Nrm(0,1)$, and set
$T_\ell$ to the nearest integer to $15000e^{V_\ell}$. The realized
edge counts have mean $15{,}191$, range $9{,}597$--$24{,}644$, and
coefficient of variation $0.199$. We then simulate the PA out-degrees at each
realized $T_\ell$. We set $\phi=0$ in this diagnostic to isolate the
marginal common-mode mechanism. Note that temporal persistence changes
sampling uncertainty but not the stationary marginal separation between
raw degrees and degree proportions.

\begin{figure}[t]\centering
\includegraphicssafe[width=\columnwidth]{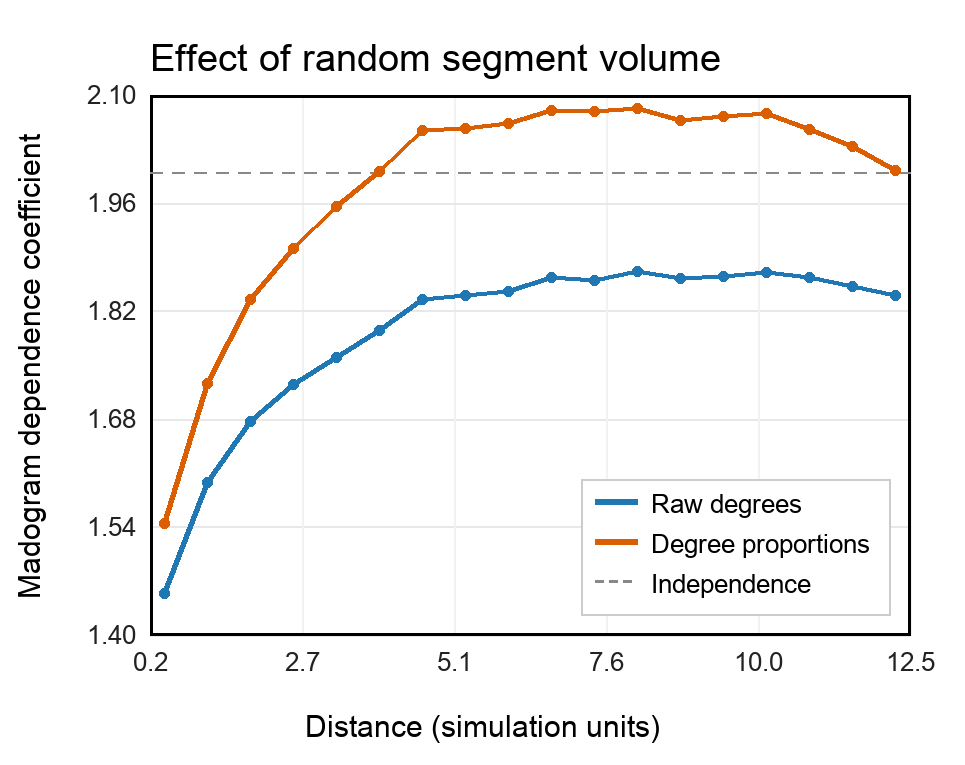}
\caption{Generative simulation of random segment volume and spatially
dependent degree proportions. Random
segment volume creates distance-independent common-mode dependence in raw PA out-degrees;
dividing by $T_\ell$ removes that common mode while retaining short-range
spatial dependence. The grey line marks the population independence
reference $\theta_F=2$.}
\label{fig:volume-sim}
\end{figure}

Fig.~\ref{fig:volume-sim} compares raw out-degrees with the degree proportions
$D_{\ell i}/T_\ell$. At distances above $8$, the raw-degree madogram
coefficient flattens at about $1.86$, whereas the degree-proportion curve reaches the
independence reference (estimated value $2.06$, which is admissible for
the finite-sample rank statistic). At distances below $2$, both
curves retain spatial dependence, with averages $1.58$ and $1.70$,
respectively. 

This experiment confirms that
spatially correlated weights determine the short-range profile, while
random $T_\ell$ creates the common mode in raw degrees.

Supplementary estimation experiments show improving MM weight recovery as the
edge horizon grows and a conservative profile estimate of $\alpha$ at moderate
horizons. The latter supports the fixed-reference treatment of $\alpha$ in the
airline analysis.

\section{Airline Network Analysis}
\label{sec:real}

The airline dataset, obtained from the U.S. Bureau of Transportation Statistics
(BTS)~\cite{btsdata}, consists of $194{,}385{,}636$ domestic flight records spanning
1987--2020 over $421$ U.S. airports. We analyze major airports,
defined as those with more than $100{,}000$ total departures and
arrivals over 1987--2020. The data are transformed into a directed
multigraph in which nodes are airports and directed edges are flight
routes. The regional panels retain flights whose origin and destination
both lie in the same division, so out-/in-degrees correspond to
intra-division departures/arrivals;
unlike undirected treatments~\cite{li2020}, this retains both
directionality and flight multiplicity. 

We define an Eastern Division
(ED) as airports between $25^\circ$N--$49^\circ$N and
$66^\circ$W--$88^\circ$W, covering most of the Eastern Time Zone and
part of the Central Time Zone (including ORD and DTW); the remaining
airports form the Other Division (OD).

\subsection{Spatial Dependence in Airport Activity}

We analyze flights from 2015--2019. Dataset~I contains $71$ airports and
$10{,}367{,}080$ intra-division flights in the Eastern Division (ED);
Dataset~II contains $61$ airports and $8{,}981{,}748$ intra-division
flights in the Other Division (OD). We divide each month into three
reporting periods: days 1--10, 11--20, and 21--end, giving $180$
observational periods. The final period lasts 8--11 days, so observed
segment volume contains a small exposure-length component in addition to
traffic intensity. For ED, dividing volume by period length changes its
coefficient of variation from $0.191$ to $0.185$; the spatial analysis uses
degree proportions, which remove this direct exposure scale. Each period
contains roughly $50{,}000$--$58{,}000$ directed edges. Geodesic distances between airports
are computed from their latitudes and longitudes.

For the ED out-degrees, the estimated madogram dependence coefficient
increases from about $1.45$ at short distances to about $1.8$ at
$2150$~km (Fig.~\ref{fig:ed}). Because smaller values indicate stronger
dependence, this pattern shows that airport activities become less
dependent as distance increases. The coefficient nevertheless remains
below the independence reference value of $2$, indicating dependence
that persists even between distant airports. A one-sided
location-permutation test confirms the increasing trend
($p=0.002$). The confidence bands are obtained from a six-period circular
moving-block bootstrap, which preserves short-term temporal dependence.

To explain the remaining dependence at large distances, we examine
variation in the log-degrees across periods. The first principal component
explains $82.6\%$ of the total variance, has positive loadings for $97\%$
of the airports, and has correlation $0.99$ with the total number of
flights per period. The latter ranges from $38{,}000$ to $80{,}000$, with
a coefficient of variation of $19\%$. Thus, the dominant fluctuation in
the raw degrees closely tracks changes in total network activity and
affects airports at all distances.

This network-wide effect is removed by replacing the raw degrees with
the out-degree proportion
\[
Q_{\ell i}^{\mathrm{out}}
=
\frac{D_{\ell i}^{\mathrm{out}}}{T_\ell}
=
\frac{D_{\ell i}^{\mathrm{out}}}
     {\sum_j D_{\ell j}^{\mathrm{out}}}.
\]
For these proportions, the madogram dependence coefficient increases to about
$1.8$ within $500$~km and levels off near the independence reference
beyond $1500$~km, where its estimate is approximately $2.06$. 

The raw-degree profile therefore combines two essential effects:
variation in total flight volume, which produces dependence at all
distances, and variation in the relative activities of individual
airports, whose dependence weakens with distance. The latter is the
short-range spatial effect described by the fitted Gaussian-process model.

\begin{figure}[t]
\centering
\includegraphicssafe[width=\columnwidth]
{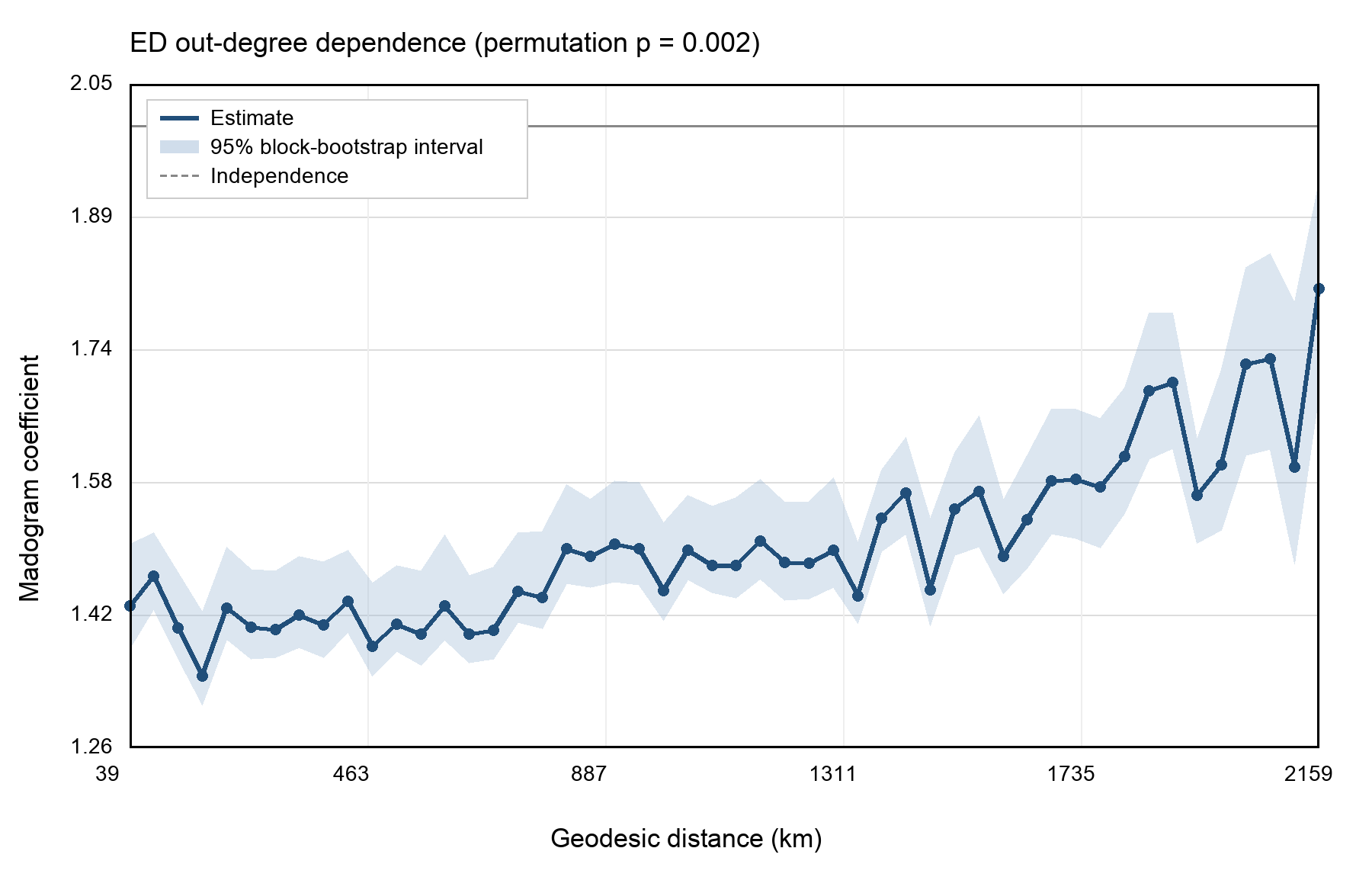}
\caption{Madogram dependence coefficient for ED out-degrees as a
function of geodesic distance, with $95\%$ confidence bands from a
six-period circular moving-block bootstrap. The coefficient increases
significantly with distance (one-sided location-permutation
$p=0.002$) but remains below the independence reference
$\theta_F=2$.}
\label{fig:ed}
\end{figure}

\subsection{Estimation Results}

The airline data provide the terminal degree counts for each period but
not a reliable ordering of the individual flights. We therefore do not
apply the ordered-history MM estimator or profile likelihood to these
data. For the out-degrees, each period contains enough edges to motivate
the large-horizon approximation in Theorem~\ref{thm:degree-tail}:
\begin{equation}
\label{eq:empirical-weight-recovery}
\log w^{\mathrm{out}}_{\ell i}
=
(1-\alpha)
\log\left(\frac{D^{\mathrm{out}}_{\ell i}}{T_\ell}\right)
+c^{\mathrm{out}}_\ell,
\end{equation}
where $c^{\mathrm{out}}_\ell$ is common to all nodes in period $\ell$.
Both $c^{\mathrm{out}}_\ell$ and the contribution of $\log T_\ell$
disappear when the common shift across nodes is removed using
\eqref{eq:contrasts}.
For the in channel, we set
$a=Q^{\mathrm{out}}_\ell$ and $p=Q^{\mathrm{in}}_\ell$ in
Proposition~\ref{prop:in-inversion}, solve
\eqref{eq:in-inverse-objective}, and construct the normalized in-weights
from \eqref{eq:in-inverse-weights}. Thus both channels respect their
respective no-self-loop limits without imposing an artificial edge order.

The calibrated attachment-rate and joint attachment--fitness checks in the
supplement place the effective exponent near $0.15$--$0.2$. We therefore use
$\alpha=0.2$ as an empirically informed reference and normalize the estimated
weights in each period so that $\sum_i w_{\ell i}=N$. Because the within-day
edge order is unavailable and terminal degree proportions do not identify
$\alpha$, this value is a reference rather than a precise estimate.
Table~\ref{tab:hubs}
reports the ten ED airports with the largest mean estimated log-weights.
ATL ranks first and ORD second, followed by CLT. The remaining positions
include several Northeast airports, including LGA, BOS, EWR, BWI, and
DCA. The out- and in-degree rankings coincide. 

\begin{table}[t]
\caption{Top ten ED airports by mean estimated log-weight at
$\alpha=0.2$. Weights are normalized within each period so that
$\sum_i w_{\ell i}=N$; the in-degree values use the exact inversion in
Proposition~\ref{prop:in-inversion}.}
\label{tab:hubs}
\centering
\scriptsize
\setlength{\tabcolsep}{1.5pt}
\begin{tabular}{@{}llcccccccccc@{}}
\toprule
$\alpha$ & & ATL & ORD & CLT & LGA & BOS & MCO & DTW & EWR & BWI & DCA\\
\midrule
$0.20$ & Out & 1.765 & 1.423 & 1.137 & 1.088 & 1.084 & 1.024 & 0.956 & 0.915 & 0.855 & 0.851\\
$0.20$ & In  & 1.859 & 1.464 & 1.153 & 1.099 & 1.095 & 1.031 & 0.959 & 0.915 & 0.853 & 0.850\\
\bottomrule
\end{tabular}
\end{table}

Repeating the out-degree analysis with the larger sensitivity value
$\alpha=0.5$ leaves $\widehat\phi$ and $\widehat\xi$ unchanged to the
reported precision. The estimated stationary variance changes by
\[
\left(\frac{1-0.5}{1-0.2}\right)^2=0.390625,
\]
as predicted by Proposition~\ref{prop:alpha-invariance}; for example,
the ED out-degree estimate changes from $\widehat\gamma=0.0274$ at
$\alpha=0.2$ to $0.0107$ at $\alpha=0.5$. Thus, across these two choices,
the estimated out-range is invariant but the variance is not. For the exact
in-degree inversion, changing $\alpha$ is not a common rescaling; nevertheless,
the fitted in-range changes only from $72.9$ to $72.2$~km for ED and from
$109.5$ to $108.5$~km for OD. Full sensitivity results are reported in the
supplement.

We next estimate temporal persistence and spatial covariance. For
numerical stability, all geodesic distances are divided by $100$ during
optimization. The estimated AR(1) coefficients are $0.934$ for both ED
degree directions and between $0.917$ and $0.920$ for OD, indicating
strong temporal persistence. After pre-whitening, the median lag-one
residual autocorrelation is $-0.05$ for ED and $0.04$ for OD, and the
corresponding 90th percentiles are $0.13$ and $0.15$, suggesting
that little lag-one dependence remains after the AR(1) adjustment.

Table~\ref{tab:region} reports the spatial parameter estimates. The
stationary variance is calculated from the innovation variance as
\[
\widehat\gamma
=
\frac{\widehat\gamma_e}{1-\widehat\phi^2}.
\]
The innovation variances are approximately $0.00350$--$0.00358$ for ED and
$0.00383$--$0.00385$ for OD. The estimated spatial range is about
$74$~km for ED and $110$~km for OD. The conditional profile intervals
are $[69.8,77.6]$~km and $[69.2,76.9]$~km for the ED out- and in-degree
analyses, respectively, and $[103.6,116.4]$~km and
$[103.3,116.0]$~km for OD. 

\begin{table}[t]
\caption{Temporal and spatial parameter estimates at $\alpha=0.2$.
The range parameter $\xi$ is reported in units of $100$~km; the in-degree
estimates use the exact no-self-loop inversion.}
\label{tab:region}
\centering
\scriptsize
\setlength{\tabcolsep}{2.5pt}
\begin{tabular}{@{}lcccccc@{}}
\toprule
Region & $\phi^{\rm out}$ & $\gamma^{\rm out}$ & $\xi^{\rm out}$ &
$\phi^{\rm in}$ & $\gamma^{\rm in}$ & $\xi^{\rm in}$\\
\midrule
ED & 0.934 & 0.0274 & 0.736 & 0.934 & 0.0281 & 0.729\\
OD & 0.917 & 0.0242 & 1.098 & 0.920 & 0.0249 & 1.095\\
\bottomrule
\end{tabular}
\end{table}

The fitted spatial range and the distance scale visible in the raw-degree
madogram have different interpretations. With
$\widehat\xi\approx74$~km, the Gaussian-copula benchmark in
Proposition~\ref{prop:latent-fmado} reaches
$\theta_{F,\mathrm{lat}}=1.95$ at approximately
\[
h=2.8\widehat\xi\approx207\ \text{km}.
\]
Thus, the fitted Gaussian process describes dependence that weakens within
a few hundred kilometers. The raw-degree madogram remains below $2$ over
much larger distances because every raw degree contains the same
period-specific total $T_\ell$. This persistent long-distance association
should therefore not be interpreted as a spatial range of several thousand
kilometers.

\subsection{Reconstruction Using Observed Flight Volume}

Note that degree proportions, $Q_{\ell i}^{\mathrm{out}}
=
\frac{D_{\ell i}^{\mathrm{out}}}{T_\ell}$, remove variation in the total number of flights,
and 
their madogram profile approaches the independence reference at large
distances, which therefore captures only the distance-dependent part of the
observed association. To determine whether variation in $T_\ell$ explains
the remaining long-distance dependence in the raw degrees, we construct
\[
D_{\ell i}^{*}
=
T_\ell^{*}Q_{\ell i}^{\mathrm{out}},
\]
where $T_\ell^{*}$ is sampled independently from the observed period
volumes. This reconstruction preserves the empirical degree proportions
but removes their observed association with total volume.

Figure~\ref{fig:modelfit} compares three profiles. The raw-degree profile
remains below the independence reference at large distances. The profile
obtained from $D_{\ell i}^{*}$ closely reproduces this behavior, whereas
the profile based on $Q_{\ell i}^{\mathrm{out}}$ approaches independence.
Hence, variation in total flight volume is sufficient to explain most of
the long-distance dependence in the ED raw degrees, while the degree
proportions retain the shorter-range dependence associated with distance.

\begin{figure}[t]
\centering
\includegraphicssafe[width=\columnwidth]{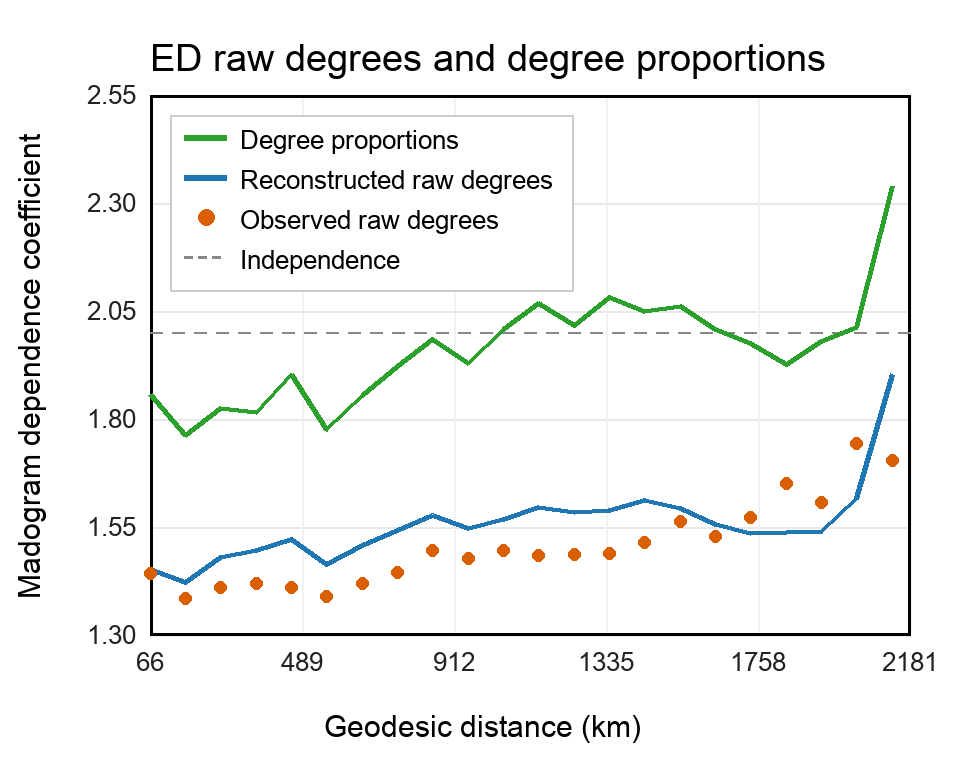}
\caption{Reconstruction of ED out-degrees using observed flight volume.
Markers show the empirical raw-degree madogram coefficient. The blue curve
is obtained from independently sampled period volumes multiplied by the
empirical degree proportions. The green curve is based on degree
proportions alone. Variation in total volume reproduces the long-distance
dependence in the raw degrees, whereas the degree proportions retain the
shorter-range dependence and approach the independence reference
$\theta_F=2$.}
\label{fig:modelfit}
\end{figure}

\subsection{Exploratory Co-exceedance Distance}
\label{sec:hub}

We then use the fitted residual covariance model to examine how the simultaneous occurrence of
unusually large degree proportions changes with distance. For the ED
out-degrees, the fitted stationary spatial parameters are
\[
(\widehat\gamma,\widehat\xi)=(0.0274,0.736),
\]
where $\widehat\xi=0.736$ corresponds to $73.6$~km because distances were
divided by $100$ during estimation. We simulate mean-zero spatial residual
fields at
the observed airport locations and transform them into limiting out-degree
proportions using \eqref{eq:transmission-map}.

Let $P_i$ denote the simulated out-degree proportion at airport $i$, and
let $q_i(u)$ be its marginal $u$-quantile. For each airport pair, define
\[
\chi(u;\bm{x}_i,\bm{x}_j)
=
\Prob\!\left\{
P_j>q_j(u)\mid P_i>q_i(u)
\right\}.
\]
We evaluate this probability at $u=0.9$ and group airport pairs by
geodesic distance, i.e.\ $\chi(0.9)$ is the probability that airport $j$
is in the upper quantile of its own degree-proportion distribution,
conditional on airport $i$ also being in its upper quantile.

For the fitted residual covariance model, the nearest-distance bin has
$\chi(0.9)\approx0.22$, compared with the independence baseline
$1-u=0.1$. Hence, conditional upper-quantile activity is approximately
$2.2$ times as likely at nearby airports as it would be under spatial
independence. The probability decreases rapidly toward $0.1$, with most
of the decline occurring by approximately $150$~km
(Fig.~\ref{fig:hub}).

As a comparison, we repeat the simulation using independent residuals with
the same marginal variance. The resulting co-exceedance probability
remains close to $0.1$ at all distances. The excess probability in the
correlated model therefore isolates the effect of spatial correlation. This
experiment does not include the estimated airport-specific mean weights.

\begin{figure}[t]
\centering
\includegraphicssafe[width=\columnwidth]{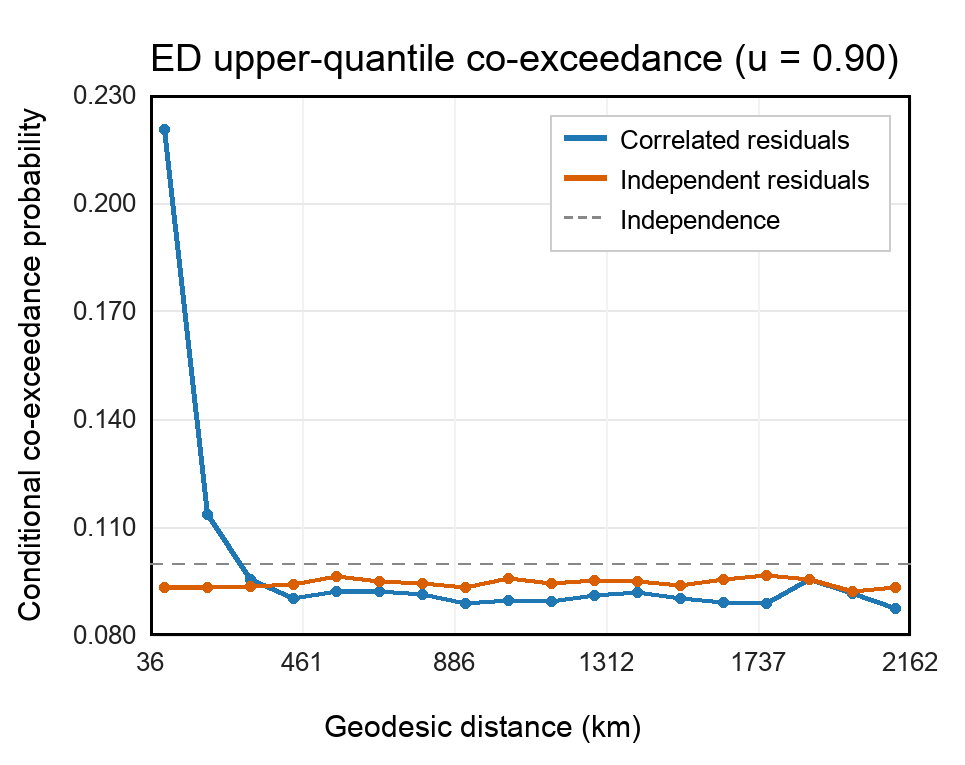}
\caption{Conditional co-exceedance probability at $u=0.9$ as a function
of geodesic distance for the ED airports. The fitted residual covariance
model is compared with independent residuals having the same marginal
variance. The grey line is the independence baseline
$1-u=0.1$. Under the correlated model, the excess co-exceedance probability
decreases sharply over approximately $150$~km.}
\label{fig:hub}
\end{figure}

\subsection{European Replication: Per-Carrier Analysis}
\label{sec:europe}

To test whether the volume/spatial decomposition generalizes beyond the
U.S. network, we apply the same diagnostics to European air traffic from
the OpenSky Network~\cite{opensky2021,openskydata}, using all flight records for 
year $2019$. Each carrier defines one directed network with airports
as nodes and flights as repeated directed edges. We use 52 complete consecutive
seven-day segments and omit the trailing day. We first retain airports in the
European International Civil Aviation Organization (ICAO) regions
\texttt{E*}/\texttt{L*} with at least three departures and three arrivals in
every segment. For each carrier, we keep up to 30 of the most active
airports and then retain only flights whose two endpoints both belong to that
set. The resulting out- and in-degrees therefore describe the same closed
directed network. 

For each carrier, we observe a similar pattern as in the U.S. case.
The raw-degree madogram dependence coefficient stays well below the
independence reference of $2$ at all distances (typically
$\theta_F\approx1.3$--$1.8$), indicating a distance-independent common
mode, while the degree-proportion coefficient rises toward $2$ with
distance, reflecting short-range spatial dependence. The first principal
component of the log-degree fluctuations explains $39$--$74\%$ of the
variance and correlates $0.83$--$0.97$ with total weekly volume, linking much
of the common mode to network-wide volume.

After $AR(1)$ pre-whitening of the log-ratio coordinates, the plug-in
Gaussian quasi-likelihood returns an interior spatial range in both directions
for Lufthansa, easyJet, Vueling, and Wizz Air
(Table~\ref{tab:europe}; Fig.~\ref{fig:europe}a). The three networks other
than Vueling contain at least fifteen airports, and their estimates span
$74$--$108$~km, which overlap the U.S. ED--OD range. British Airways and Air
France have an interior out-range but a boundary in-range, while both channels
are unidentified for KLM and Ryanair. These mixed results support replication
of the short-range component in several carriers without implying universality.
We also use the reference exponent $\alpha=0.2$, and by
Proposition~\ref{prop:alpha-invariance}, the out-channel range is invariant to
this reconstruction choice. The exact in-degree inversion has no corresponding
algebraic invariance, so its estimates are conditional on $\alpha=0.2$.

\begin{table}[h]
\caption{Per-carrier European estimates (52 complete seven-day segments,
$2019$). $N$ is the number of airports; PC1 is the common-mode variance
fraction; $\widehat\xi$ is the pre-whitened spatial range for the out-/
in-degree channels (``bd.'' marks a boundary, i.e.\ unidentified,
estimate). Hub and low-cost point-to-point (LCC) carriers are
distinguished.}
\label{tab:europe}
\centering
\scriptsize
\setlength{\tabcolsep}{3.5pt}
\begin{tabular}{@{}llccc@{}}
\toprule
Carrier & Type & $N$ & PC1 & $\widehat\xi$ out/in (km)\\
\midrule
Lufthansa       & hub & 30 & $39\%$ & $108$ / $74$\\
British Airways & hub & 30 & $46\%$ & $139$ / bd.\\
Air France      & hub & 25 & $46\%$ & $142$ / bd.\\
KLM             & hub & 30 & $53\%$ & bd. / bd.\\
Ryanair         & LCC & 30 & $59\%$ & bd. / bd.\\
easyJet         & LCC & 29 & $74\%$ & $99$ / $106$\\
Vueling         & LCC & 12 & $61\%$ & $21$ / $65$\\
Wizz Air        & LCC & 17 & $73\%$ & $88$ / $99$\\
\bottomrule
\end{tabular}
\end{table}

\begin{figure}[t]\centering
\includegraphicssafe[width=\columnwidth]{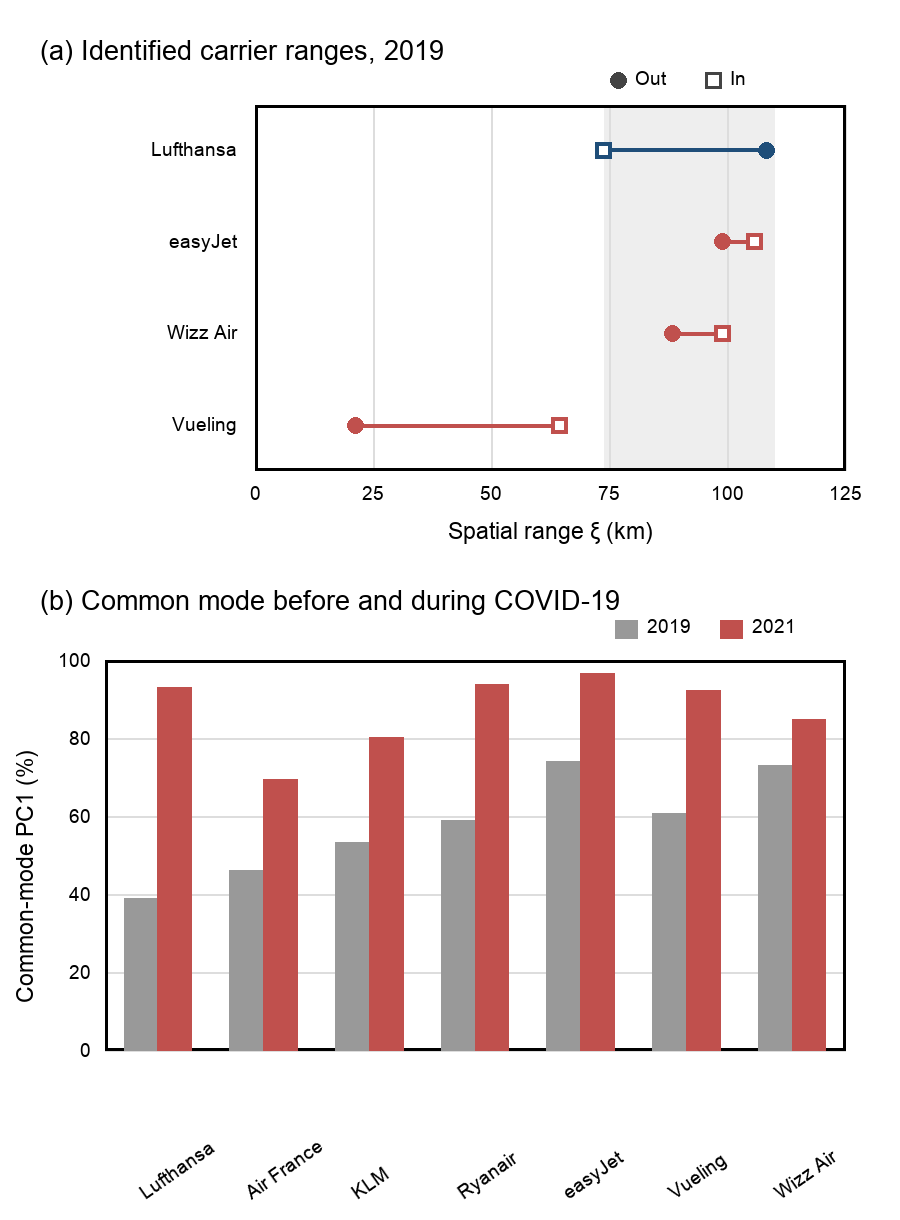}
\caption{Per-carrier European analysis. (a) Pre-whitened spatial range for the
four carriers with both channels identified (filled: out-degree, open:
in-degree; blue: hub, red: low-cost), with the shaded band marking the U.S.
ED--OD range ($74$--$110$~km). (b) Common-mode variance fraction (leading
principal component) in $2019$ versus the pandemic year $2021$. The same
airport-selection rule is applied separately to each year, so panel (b) is a
qualitative stress comparison.}
\label{fig:europe}
\end{figure}

Applying the same selection rule separately to the pandemic year $2021$
provides a temporal stress test of the common-mode mechanism
(Fig.~\ref{fig:europe}b). Network-wide volume
swings are far larger (the segment-volume coefficient of variation rises
from $0.07$--$0.26$ to $0.27$--$0.67$), and the common-mode fraction rises
accordingly, reaching $69$--$97\%$ with raw $\theta_F$ as low as $1.14$
as whole fleets were grounded and restarted together. For Ryanair the first
principal component reaches $94\%$ and correlates $1.00$ with volume. British
Airways has no airport set satisfying the regular-service criterion in $2021$,
so is not reported.

The carrier-level analysis reproduces the volume/common-mode and short-range
spatial decomposition in a separate geographic dataset. The identified ranges
are of the same order as the U.S. estimates, although the carrier-level point
estimates remain exploratory.

\section{Concluding Remarks}
\label{sec:concl}

We proposed a directed preferential attachment model with spatially
correlated, time-varying node weights. Under sublinear attachment,
out-degree proportions converge to powered and normalized out-weights,
whereas self-loop exclusion yields a coupled fixed-point limit for the
in-degree proportions. At $\alpha=1$, unique maximal weights produce
winner-take-all behavior. The proposed inference framework combines a
strictly concave in-weight inverse for terminal counts, a globally convergent MM estimator,
profile likelihood for $\alpha$, AR(1) pre-whitening, and spatial covariance
estimation. Simulations confirm that
correlated weights produce dependence that weakens with distance and also
show that $\alpha$ can be weakly identified at moderate horizons.

Empirically, degree co-movement in the U.S. airline network decomposes
into distance-independent common-mode dependence carried by observed network-wide volume and a
short-range spatial component, with pre-whitened ranges of about
$74$~km (ED) and $110$~km (OD). A semiparametric reconstruction reproduces
the common mode, resolving the gap between the latent covariance range and the
raw madogram scale. The fitted residual covariance yields an exploratory
co-exceedance scale of roughly $150$~km. A per-carrier analysis of European
air traffic (Section~\ref{sec:europe}) reproduces the same decomposition
and yields comparable interior spatial ranges in three carrier networks with
at least fifteen airports.

Since airline flows mediate long-range disease importation, the limiting
edge frequencies could serve as a mobility matrix in a meta-population
epidemic model~\cite{colizza2006,brockmann2013}. The present analysis,
however, models flight activity rather than disease transmission, so
airport co-exceedance does not directly imply synchronized outbreaks.
Embedding the fitted network in a compartmental transmission model is a
natural next step. Other extensions include nonstationary or anisotropic
covariance functions, correlated out- and in-weight fields, and scalable
inference for growing networks.

\section*{Acknowledgment}
The authors thank Yanchao Liu for conducting the preliminary exploratory
study of spatial dependence in airline networks that motivated this work.

\clearpage
\setcounter{section}{0}
\setcounter{equation}{0}
\setcounter{table}{0}
\setcounter{figure}{0}
\setcounter{theorem}{0}
\setcounter{lemma}{0}
\setcounter{proposition}{0}
\setcounter{corollary}{0}
\setcounter{definition}{0}
\setcounter{algorithm}{0}
\renewcommand{\thesection}{S\arabic{section}}
\renewcommand{\theequation}{S\arabic{equation}}
\renewcommand{\thetable}{S\arabic{table}}
\renewcommand{\thefigure}{S\arabic{figure}}
\renewcommand{\thetheorem}{S\arabic{theorem}}
\renewcommand{\thelemma}{S\arabic{lemma}}
\renewcommand{\theproposition}{S\arabic{proposition}}
\renewcommand{\thecorollary}{S\arabic{corollary}}
\renewcommand{\thedefinition}{S\arabic{definition}}
\renewcommand{\thealgorithm}{S\arabic{algorithm}}

\begin{center}
{\Large\bfseries Supplementary Material}\\[2pt]
{\itshape Equation, Proposition, Theorem, and Algorithm numbers
without the ``S'' prefix refer to the main text above.}
\end{center}

\section{Deferred Proofs}
\label{sec:proofs}

We restate, for self-containedness, the results proved here. Theorem,
Proposition, Equation, and Assumption numbers without the ``S'' prefix
refer to the main paper.

\subsection{Proof of Theorem 1 (Out- and In-Degree-Proportion Limits)}
We begin with the out-degree process. Conditional on the
history $\mathcal F_T$, the probability that the next source node is $i$ is
\[
q_i(Y(T))=\frac{\wout_i\{\Dout_i(T)+1\}^{\alpha}}
{\sum_{k=1}^{N}\wout_k\{\Dout_k(T)+1\}^{\alpha}}.
\]
Since $Y_i^{\mathrm{out}}(T)=(\Dout_i(T)+1)/(T+N)$, the common factor
$(T+N)^\alpha$ cancels, and therefore
\[
q_i(y)=\frac{\wout_i y_i^\alpha}{\sum_{k=1}^{N}\wout_k y_k^\alpha},
\qquad y\in\Delta_N,
\]
where $\Delta_N=\{y\in\mathbb R_+^N:\sum_{i=1}^{N}y_i=1\}$. The normalized
process satisfies the stochastic approximation recursion
\[
Y^{\mathrm{out}}(T+1)=Y^{\mathrm{out}}(T)+\frac{1}{T+N+1}
\{e_{I_{T+1}^{(1)}}-Y^{\mathrm{out}}(T)\},
\]
where $e_i$ is the $i$th unit vector. Taking conditional expectation gives
the mean-field drift
\[
\mathbb E[Y^{\mathrm{out}}(T+1)-Y^{\mathrm{out}}(T)\mid\mathcal F_T]
=\frac{q(Y^{\mathrm{out}}(T))-Y^{\mathrm{out}}(T)}{T+N+1}.
\]
Thus the limiting ODE is $\dot y=q(y)-y$.

Let $r=1/(1-\alpha)$ and
\[
p_i=\frac{(\wout_i)^r}{\sum_{k=1}^{N}(\wout_k)^r}.
\]
Then $\wout_i=Cp_i^{1-\alpha}$ for some $C>0$, and substituting into
$q_i(p)$ gives $q_i(p)=p_i$. Conversely, if $y=q(y)$ and $y_i>0$, then
$y_i^{1-\alpha}\propto\wout_i$, so normalization gives $y=p$. Hence the
interior equilibrium is unique.

It remains to show global attraction. Let
$V(y)=\sum_{i=1}^{N}p_i\log(p_i/y_i)$, put
$\beta=1-\alpha\in(0,1)$, and set $R_i=p_i/y_i$. Using
$\wout_i=Cp_i^\beta$,
\[
q_i(y)=\frac{p_i^\beta y_i^{1-\beta}}
{\sum_{k=1}^{N}p_k^\beta y_k^{1-\beta}}
=\frac{y_iR_i^\beta}{\sum_{k=1}^{N}y_kR_k^\beta}.
\]
Along the ODE,
\[
\frac{d}{dt}V(y(t))=1-
\frac{\sum_{i=1}^{N}y_iR_i^{\beta+1}}
{\sum_{i=1}^{N}y_iR_i^\beta}.
\]
Since $\sum_i y_iR_i=1$ and $x\mapsto x^\beta$ is increasing,
Chebyshev's sum inequality gives
$\sum_i y_iR_i^{\beta+1}\ge(\sum_i y_iR_i)(\sum_i y_iR_i^\beta)$.
Therefore $dV/dt\le0$, with equality only when all $R_i$ are equal,
which, on the simplex, implies $y=p$. Thus $V$ is a strict Lyapunov
function and $p$ is globally asymptotically stable.

To verify that the boundary equilibria do not attract this recursion,
note that, on the simplex,
\[
q_i(y)\ge
\frac{\wout_i}{(\max_k\wout_k)N^{1-\alpha}}y_i^\alpha.
\]
Since $\alpha<1$, $q_i(y)-y_i>0$ uniformly whenever $y_i$ is sufficiently
small. Thus every boundary face is repelling from the interior. Applying
the usual stopping-time truncation on
$\Delta_{N,\varepsilon}=\{y:\min_i y_i\ge\varepsilon\}$ and then letting
$\varepsilon\downarrow0$ rules out a boundary limit; on every such
truncated simplex the drift is Lipschitz. The stochastic approximation
recursion has step sizes whose sum diverges and whose squared sum is
finite, and its martingale-difference noise is bounded. The ODE method,
together with the strict Lyapunov function above and boundary avoidance,
then implies
$Y^{\mathrm{out}}(T)\to p^{\mathrm{out}}$ a.s. Since
\[
\Dout_i(T)/T=((T+N)/T)Y_i^{\mathrm{out}}(T)-1/T,
\]
we obtain $\Dout_i(T)/T\to p_i^{\mathrm{out}}$ a.s. Indicator convergence
for thresholds $u\neq p_i^{\mathrm{out}}$ and finite averaging over
$i=1,\ldots,N$ yield the convergence of $\overline F_T^{\mathrm{out}}(u)$.
Finally, if $\log\wout$ is Gaussian with covariance $\Sigma$, then
$\log\{(\wout_i)^{1/(1-\alpha)}\}=(1-\alpha)^{-1}\log\wout_i$, so the
powered latent field remains lognormal with log-covariance
$(1-\alpha)^{-2}\Sigma$.

We now prove the in-degree statement. Write
$X(T)=Y^{\mathrm{out}}(T)$ and $Y(T)=Y^{\mathrm{in}}(T)$, and define
the source-selection probability
\[
s_i(x)=\frac{\wout_i x_i^\alpha}
{\sum_{a=1}^{N}\wout_a x_a^\alpha}.
\]
For $y$ in the interior of the simplex, put
\[
B_i(y)=\sum_{k\ne i}\win_k y_k^\alpha
\]
and let the marginal probability that the next target is $j$ be
\begin{equation}
Q_j(x,y)=\win_jy_j^\alpha
\sum_{i\ne j}\frac{s_i(x)}{B_i(y)}.
\label{eq:target-drift}
\end{equation}
Indeed, this expression averages the target probability conditional on
source $i$ over the source distribution $s(x)$. The in-degree proportion
process obeys
\[
Y(T+1)=Y(T)+\frac{1}{T+N+1}
\{e_{I_{T+1}^{(2)}}-Y(T)\}.
\]
Its conditional mean drift is therefore
$Q\{X(T),Y(T)\}-Y(T)$. Since the out-degree result gives
$X(T)\to p^{\mathrm{out}}$ and
$s(p^{\mathrm{out}})=p^{\mathrm{out}}$, the limiting target ODE is
\begin{equation}
\dot y=Q(p^{\mathrm{out}},y)-y.
\label{eq:target-ode}
\end{equation}

Consider the potential
\begin{equation}
\Phi(y)=\sum_{i=1}^{N}p_i^{\mathrm{out}}\log B_i(y).
\label{eq:target-potential}
\end{equation}
Because $0<\alpha<1$, each coordinate map $y_k\mapsto y_k^\alpha$ is
strictly concave. Each term in \eqref{eq:target-potential} is concave,
and for any two distinct simplex points at least one term is strictly
concave: two coordinates must differ and, since $N\ge3$, some excluded
index differs from both. As every $p_i^{\mathrm{out}}>0$, $\Phi$ is
strictly concave on $\Delta_N$. Its maximizer is interior. At a vertex
one $B_i$ vanishes, while at any other boundary point the one-sided
derivative in a missing coordinate is infinite because
$\alpha-1<0$.

For an interior point,
\[
\frac{\partial\Phi(y)}{\partial y_j}
=\alpha\win_jy_j^{\alpha-1}
\sum_{i\ne j}\frac{p_i^{\mathrm{out}}}{B_i(y)}
=\alpha\frac{Q_j(p^{\mathrm{out}},y)}{y_j}.
\]
At the constrained maximizer this derivative equals a common Lagrange
multiplier $\lambda$. Multiplying by $y_j$, summing over $j$, and using
$\sum_jQ_j(p^{\mathrm{out}},y)=1$ gives $\lambda=\alpha$. Hence the
first-order conditions are exactly
$Q_j(p^{\mathrm{out}},y)=y_j$ for all $j$, which is the fixed-point
equation in the main paper. Strict concavity gives existence and
uniqueness of the interior fixed point $p^{\mathrm{in}}$.

The same potential proves global attraction. Along
\eqref{eq:target-ode},
\begin{align*}
\frac{d}{dt}\Phi(y(t))
&=\alpha\sum_{j=1}^{N}\frac{Q_j(p^{\mathrm{out}},y)}{y_j}
\{Q_j(p^{\mathrm{out}},y)-y_j\}\\
&=\alpha\left\{\sum_{j=1}^{N}
\frac{Q_j(p^{\mathrm{out}},y)^2}{y_j}-1\right\}\ge0,
\end{align*}
where the final inequality is Cauchy--Schwarz. Equality holds only when
$Q_j/y_j$ is constant over $j$, and the simplex constraints then imply
$Q=y$. Thus $\Phi$ is a strict Lyapunov function for the unique target
equilibrium.

Boundary faces are repelling. Since
$B_i(y)\le(\max_k\win_k)N^{1-\alpha}$,
\[
Q_j(p^{\mathrm{out}},y)\ge
\frac{\win_j(1-p_j^{\mathrm{out}})}
{(\max_k\win_k)N^{1-\alpha}}y_j^\alpha,
\]
so $Q_j-y_j>0$ whenever $y_j$ is sufficiently small. The same bound,
with a slightly smaller positive constant, holds once $X(T)$ is close to
$p^{\mathrm{out}}$. On compact subsets of the interior the drift is
Lipschitz, the martingale-difference noise is bounded, and the additional
drift perturbation caused by $X(T)-p^{\mathrm{out}}$ vanishes. The ODE
method for asymptotically autonomous stochastic approximation~\cite{benaim1999S}, together with boundary avoidance and the globally
attracting equilibrium above, yields
$Y^{\mathrm{in}}(T)\to p^{\mathrm{in}}$ a.s. Finally,
$\Din_j(T)/T=((T+N)/T)Y_j^{\mathrm{in}}(T)-1/T$, proving the claimed
in-degree-proportion limit for $0<\alpha<1$.

We finish with $\alpha=1$. For the source process, use the standard
continuous-time embedding. Let $Z_i(u)$, $i=1,\ldots,N$, be independent
Yule processes, each starting from one individual, in which every
individual of type $i$ gives birth at rate $\wout_i$. Conditional on the
current vector $Z(u)=z$, the next birth has type $i$ with probability
\[
\frac{\wout_i z_i}{\sum_k\wout_k z_k}.
\]
Consequently the vector observed at successive birth times has exactly
the law of $(\Dout_i(T)+1)_{i=1}^N$ under linear reinforcement.

For each $i$, the nonnegative Yule martingale satisfies
\[
e^{-\wout_i u}Z_i(u)\longrightarrow W_i\qquad\text{a.s.},
\]
where the $W_i$ are independent $\operatorname{Exp}(1)$ variables and
are therefore strictly positive almost surely. If $i_*$ is the unique
maximizer of the out-weight, then for every $i\ne i_*$,
\[
\frac{Z_i(u)}{Z_{i_*}(u)}
=e^{-(\wout_{i_*}-\wout_i)u}
\frac{e^{-\wout_i u}Z_i(u)}
     {e^{-\wout_{i_*}u}Z_{i_*}(u)}\longrightarrow0
\quad\text{a.s.}
\]
Evaluating at the successive birth times, which diverge almost surely,
gives $Y^{\mathrm{out}}(T)\to e_{i_*}$. More generally, if
$M=\argmax_i\wout_i$, the same calculation gives
\[
Y_i^{\mathrm{out}}(T)\longrightarrow
\mathbf1\{i\in M\}\frac{W_i}{\sum_{k\in M}W_k}.
\]
Normalized independent exponential variables have the
$\operatorname{Dirichlet}(1,\ldots,1)$ law, proving the tie statement.

Now assume $i_*$ is unique and let
$j_*=\argmax_{j\ne i_*}\win_j$ be unique. Source convergence implies
\[
\frac1T\sum_{i\ne i_*}\Dout_i(T)\longrightarrow0.
\]
The embedding gives a rate as well. If
$\rho=\max_{i\ne i_*}\wout_i/\wout_{i_*}<1$, then, for every sufficiently
small $\varepsilon>0$, the fraction of non-$i_*$ source events is
$O(T^{\rho-1+\varepsilon})$ almost surely. Consequently its contribution
to the normalized target recursion is summable because
\[
\sum_{T=1}^{\infty}T^{-1}T^{\rho-1+\varepsilon}<\infty
\qquad(\rho+\varepsilon<1).
\]
A target equal to $i_*$ can occur only when the source differs from
$i_*$; hence $\Din_{i_*}(T)/T\to0$. On the remaining target set
$H=V\setminus\{i_*\}$, all but $o(T)$ edge events have source $i_*$.
The normalized target process on $H$ is therefore a stochastic
approximation with a summable drift perturbation and limiting choice map
\[
q_j^H(y)=\frac{\win_jy_j}{\sum_{k\in H}\win_ky_k},
\qquad j\in H.
\]
Its limiting ODE is $\dot y=q^H(y)-y$. For two positive coordinates
$j,k\in H$, every interior ODE trajectory satisfies
\[
\frac{d}{du}\log\frac{y_j(u)}{y_k(u)}
=\frac{\win_j-\win_k}{\sum_{a\in H}\win_a y_a(u)}.
\]
The denominator is at most $\max_{a\in H}\win_a$. Therefore, for every
$j\ne j_*$, the ratio $y_{j_*}(u)/y_j(u)$ grows at least exponentially,
and every interior trajectory converges to $e_{j_*}$. The actual urn has
strictly positive initial counts in every coordinate. The ODE method for
stochastic approximation with summable perturbations~\cite{benaim1999S}, together with nonconvergence to the linearly unstable
boundary equilibria~\cite{pemantle1990S}, therefore gives convergence to
the unique attracting vertex: $Y^{\mathrm{in}}(T)\to e_{j_*}$ almost
surely. This proves the linear-case assertion. \hfill$\blacksquare$

\subsection{Proof of Theorem 2 (Edge-Measure Limit)}
Let $\mathcal F_t$ denote the history up to time $t$. For $i\neq j$,
define
\[
\pi_{ij}(t)=\mathbb P\{I_t^{(1)}=i,I_t^{(2)}=j\mid\mathcal F_{t-1}\}.
\]
By the model definition,
\begin{align*}
\pi_{ij}(t)
&=\frac{\wout_i\{\Dout_i(t-1)+1\}^{\alpha}}
{\sum_{a=1}^{N}\wout_a\{\Dout_a(t-1)+1\}^{\alpha}}\\
&\quad\times
\frac{\win_j\{\Din_j(t-1)+1\}^{\alpha}}
{\sum_{b\neq i}\win_b\{\Din_b(t-1)+1\}^{\alpha}}.
\end{align*}
Because $(D_i^A(t-1)+1)/t\to p_i^A$ for
$A\in\{\mathrm{out},\mathrm{in}\}$ and all weights are positive,
continuity gives $\pi_{ij}(t)\to K_{ij}$ a.s. Now write
\[
A_{ij}(T)=\sum_{t=1}^{T}\mathbf{1}\{I_t^{(1)}=i,I_t^{(2)}=j\}.
\]
The difference
$M_{ij}(T)=A_{ij}(T)-\sum_{t=1}^{T}\pi_{ij}(t)$ is a martingale with
bounded increments. Hence, by the martingale strong law,
$M_{ij}(T)/T\to0$ a.s. Since Cesaro convergence gives
$T^{-1}\sum_{t=1}^{T}\pi_{ij}(t)\to K_{ij}$, it follows that
$A_{ij}(T)/T\to K_{ij}$ a.s. for each ordered pair. There are finitely
many ordered pairs, so the convergence holds simultaneously for all
$i\neq j$, and weak convergence of the finite atomic measures follows.
For $\alpha=1$ under unique maximizers, every edge other than
$i_*\to j_*$ either has source different from $i_*$ or target different
from $j_*$. Therefore
\begin{align*}
0\le1-\frac{A_{i_*j_*}(T)}{T}
&\le\frac{T-\Dout_{i_*}(T)}{T}
+\frac{T-\Din_{j_*}(T)}{T}\\
&\longrightarrow0\qquad\text{a.s.}
\end{align*}
All other edge frequencies vanish, proving convergence to the stated
point mass.
\hfill$\blacksquare$

\subsection{Proof of Theorem 3 (Spatial-Dependence Transfer)}
By Theorem~1 of the main paper, for every node $i$,
\[
X_i^{(T)}=\frac{D_i^{\mathrm{out}}(T)}{T}\to
\frac{(W_i^{\mathrm{out}})^r}{\sum_{k=1}^{N}(W_k^{\mathrm{out}})^r}
=P_i\qquad \text{a.s.}
\]
Therefore $X^{(T)}\to P$ a.s., and in particular
$(X_i^{(T)},X_j^{(T)})\to(P_i,P_j)$ a.s.
Let $F_{i,T}^{\mathrm{deg}}$ be the marginal distribution of
$X_i^{(T)}$, and let $F_i^P$ be that of $P_i$. Since
$X_i^{(T)}\Rightarrow P_i$ and $F_i^P$ is continuous, Polya's theorem
gives uniform convergence of the marginal distribution functions:
$\sup_x|F_{i,T}^{\mathrm{deg}}(x)-F_i^P(x)|\to0$. Consequently,
\begin{align*}
&|F_{i,T}^{\mathrm{deg}}(X_i^{(T)})-F_i^P(P_i)|\\
&\quad\le
\sup_x|F_{i,T}^{\mathrm{deg}}(x)-F_i^P(x)|
+|F_i^P(X_i^{(T)})-F_i^P(P_i)|\to0
\end{align*}
almost surely, where the last term vanishes by continuity of $F_i^P$.
The same holds for node $j$. The transformed absolute difference
therefore converges almost surely and is bounded by one, so dominated
convergence gives
\[
\frac12\mathbb E|F_{i,T}^{\mathrm{deg}}(X_i^{(T)})-
F_{j,T}^{\mathrm{deg}}(X_j^{(T)})|
\to
\frac12\mathbb E|F_i^P(P_i)-F_j^P(P_j)|.
\]
The madogram-coefficient convergence follows from the continuous
transformation $\theta_F=(1+2v_F)/(1-2v_F)$ when $v_F<1/2$.
\hfill$\blacksquare$

\subsection{Proof of Proposition 1 (Gaussian-Copula F-Madogram)}
Because $g$ is strictly increasing, the probability integral transforms
$U_i=F\{Z(\bm{x}_i)\}$ have a Gaussian copula with correlation
$\rho=\rho(h)$. Using
$\max(U_1,U_2)=\tfrac12(U_1+U_2)+\tfrac12|U_1-U_2|$,
\[
v_F=\tfrac12\E|U_1-U_2|
=\E\{\max(U_1,U_2)\}-\tfrac12.
\]
If $C$ is the Gaussian copula distribution function, then
\[
\E\{\max(U_1,U_2)\}=1-\int_0^1 C(t,t)\,dt.
\]
Substituting $t=\Phi(z)$ and introducing an independent
$Z_3\sim\Nrm(0,1)$ gives
\[
\int_0^1 C(t,t)\,dt
=\Prob\{Z_1\le Z_3,\;Z_2\le Z_3\}.
\]
The vector $(Z_1-Z_3,Z_2-Z_3)$ is centered bivariate normal with equal
variances and correlation $(1+\rho)/2$. Sheppard's orthant formula then
gives
\[
\Prob\{Z_1-Z_3\le0,\;Z_2-Z_3\le0\}
=\frac14+\frac{1}{2\pi}\arcsin\!\left(\frac{1+\rho}{2}\right),
\]
which proves the Gaussian-copula F-madogram proposition in the main
paper; its transformed dependence-coefficient form follows immediately
from the stated F-madogram mapping.
\hfill$\blacksquare$

\subsection{Proof of Proposition 3 (Exact In-Degree Inversion)}

For $v\in\R^N$, define the conditional target probabilities
\[
\pi_{ij}(v)=
\frac{\mathbf 1\{j\ne i\}e^{v_j}}
{\sum_{k\ne i}e^{v_k}},
\qquad
q_j(v)=\sum_{i=1}^{N}a_i\pi_{ij}(v).
\]
The gradient of the criterion in Proposition~3 of the main paper is
\[
\nabla\mathcal J_{a,p}(v)=p-q(v),
\]
and its negative Hessian is
\[
-\nabla^2\mathcal J_{a,p}(v)
=\sum_{i=1}^{N}a_i
\{\operatorname{diag}(\pi_i)-\pi_i\pi_i^\top\}.
\]
Thus, for any $h\in\R^N$,
\[
-h^\top\nabla^2\mathcal J_{a,p}(v)h
=\sum_{i=1}^{N}a_i\operatorname{Var}_{\pi_i(v)}(h_j)\ge0.
\]
Equality requires $h$ to be constant on every set
$\{1,\ldots,N\}\setminus\{i\}$. Since $N\ge3$, these conditions imply
$h=c\onev$. Hence the objective is strictly concave on
$\onev^\perp$.

It remains to establish existence. For nonzero $d\in\onev^\perp$,
\[
\lim_{t\to\infty}\frac{\mathcal J_{a,p}(td)}{t}
=p^\top d-\sum_{i=1}^{N}a_i\max_{k\ne i}d_k.
\]
If the maximum coordinate $d_j$ is unique and $d_{(2)}$ is the
second-largest coordinate, then
\[
\sum_i a_i\max_{k\ne i}d_k
=(1-a_j)d_j+a_jd_{(2)},
\]
whereas
\[
p^\top d\le p_jd_j+(1-p_j)d_{(2)}.
\]
Their difference is at least
$(1-a_j-p_j)(d_j-d_{(2)})>0$. If the maximum is attained at least twice,
the second term equals $\max_jd_j$, while positivity of $p$ and
nonconstancy of $d$ give $p^\top d<\max_jd_j$. Therefore the displayed
limit is negative in every nonzero direction in $\onev^\perp$.
The limiting slope is continuous on the unit sphere in $\onev^\perp$;
compactness therefore gives a uniform negative upper bound. Hence
$\mathcal J_{a,p}(v)\to-\infty$ as $\|v\|\to\infty$ within that subspace,
so it attains a unique maximum $\widehat v$.

Because $\mathcal J_{a,p}(v+c\onev)=\mathcal J_{a,p}(v)$ and the entries
of its gradient sum to zero, the constrained first-order condition at
$\widehat v$ is $p=q(\widehat v)$. Set
$z_j=e^{\widehat v_j}$ and $w_j=z_jp_j^{-\alpha}$. Then
\begin{align*}
q_j(\widehat v)
&=z_j\sum_{i\ne j}\frac{a_i}{\sum_{k\ne i}z_k}\\
&=w_jp_j^\alpha
\sum_{i\ne j}\frac{a_i}
{\sum_{k\ne i}w_kp_k^\alpha},
\end{align*}
which is exactly the target fixed-point equation in the main paper.
Conversely, any positive weights producing $p$ make
$v_j=\log(w_jp_j^\alpha)$ a stationary point after centering. Strict
concavity makes this point unique modulo a common shift, proving the
proportionality claim. Normalizing the weights to sum to $N$ gives the
formula stated in Proposition~3. \hfill$\blacksquare$

\subsection{Proof of Theorem 4 (Global Optimality in Log-Weights)}
We use the following facts.

\begin{lemma}[Log-sum-exp convexity~\cite{boyd2004S}]
\label{lem:lse}
$f(\bm{x})=\log\sum_{i=1}^{n}e^{x_i}$ is convex on $\R^n$ with Hessian
$\nabla^2 f(\bm{x})=\frac{1}{(\onev^\top z)^2}\bigl((\onev^\top z)
\,\mathrm{diag}(z)-zz^\top\bigr)$, $z=(e^{x_1},\dots,e^{x_n})$.
\end{lemma}

\begin{lemma}[Generalized version]
\label{lem:glse}
For $\lambda=(\lambda_1,\dots,\lambda_n)\neq0$, $\lambda_i\ge0$,
$f(\bm{x})=\log\sum_i\lambda_i e^{x_i}$ is convex on $\R^n$. If all
$\lambda_i>0$, the null vectors $v\neq0$ with
$v^\top\nabla^2 f(\bm{x})v=0$ are exactly $v=c\onev$,
$c\in\R\setminus\{0\}$; if some $\lambda_k=0$, the null space is
$\{v:v_i=c\ \forall i\in I;\ v_j\ \text{free}\ \forall j\notin I\}$
with $I=\{i:\lambda_i>0\}$, by direct differentiation and the
Cauchy--Schwarz equality condition.
\end{lemma}

\begin{lemma}[First-order convexity condition~\cite{boyd2004S}]
\label{lem:foc}
A differentiable $f:\R^n\to\R$ is convex iff $f(y)\ge f(x)+\nabla
f(x)^\top(y-x)$ for all $x,y$.
\end{lemma}

Substituting $z^{(1)}=\log\wout$ into the out-weight log-likelihood of
the main paper gives
\begin{align*}
\ell(z^{(1)})
&=\sum_i\Dout_i(T)z^{(1)}_i\\
&\quad-\sum_{t=1}^{T}\log\!\sum_{k\in[N]}
(\Dout_k(t-1)+1)^{\alpha}e^{z^{(1)}_k}+A.
\end{align*}
Since $(\Dout_k(t-1)+1)^{\alpha}>0$, Lemma~\ref{lem:glse}
(case $\lambda>0$) gives each $\nabla^2 f^{(1)}_t\succeq0$, hence
$\nabla^2\ell(z^{(1)})=-\sum_t\nabla^2 f^{(1)}_t\preceq0$, so $\ell$ is
concave. Under Assumption~1 of the main paper, every terminal count is
positive, so on the normalized simplex the objective tends to $-\infty$
when any weight approaches zero. It therefore attains its maximum in the
interior. By Lemma~\ref{lem:foc}, any $z^{(1)}_*$ with
$\nabla\ell(z^{(1)}_*)=0$ is a global maximizer; the transformation
$\wout\mapsto z^{(1)}$ is one-to-one, so the corresponding $\wout_*$
optimizes $\ell(\wout)$ and solves the likelihood equations. For the
in-degree case the summation index $k\neq I^{(1)}_t$ removes one term,
so the relevant coefficient vanishes; Lemma~\ref{lem:glse} (case
$\lambda\neq0$, $\lambda_i\ge0$) still yields concavity and the same
conclusion. Non-uniqueness: if $(\wout_*,\win_*)$ solves the
equations, so does $(c_1\wout_*,c_2\win_*)$, $c_1,c_2\in\R_+$, by
homogeneity. A second-order Taylor expansion with Lagrange remainder,
together with the null-space characterization in
Lemma~\ref{lem:glse} and Assumption~1 of the main paper, shows all optima
differ by an additive constant in log-scale, i.e.\ by a positive
multiplicative constant in the original scale. Constraining
$\|\bm{w}\|_1$ then yields uniqueness. \hfill$\blacksquare$

\subsection{Proof of Proposition 4 (MM Minorizer)}
Throughout this proof write, for the out-degree case,
\begin{equation}
a_{i,t}:=(\Dout_i(t-1)+1)^{\alpha}>0,
\label{eq:ashort}
\end{equation}
(the in-degree case is identical with $\Dout$ replaced by $\Din$ and
the index restricted to $i\neq I^{(1)}_t$). Proposition~4 of the
main paper states that, at iteration $s$, the surrogate
\begin{align}
g_1\bigl(\wout\mid\wout(s)\bigr)
&=\sum_{i\in[N]}\Dout_i(T)\log\wout_i\nonumber\\
&\quad-\sum_{t=1}^{T}\log\Bigl(\sum_{i\in[N]}\wout_i(s)\,a_{i,t}\Bigr)
\nonumber\\
&\quad-\sum_{t=1}^{T}\frac{\sum_{i}\wout_i\,a_{i,t}}
{\sum_{i}\wout_i(s)\,a_{i,t}}+T
\label{eq:g1}
\end{align}
minorizes $\ell(\wout)$ and is tangent at $\wout(s)$, with an
analogous statement for $g_2(\win\mid\win(s))$ over the restricted
index set $i\neq I^{(1)}_t$.

\emph{Tangency.} Substituting $\wout=\wout(s)$ into \eqref{eq:g1}, the
ratio in the third term equals $1$ for every $t$, so that term equals
$-\sum_{t=1}^{T}1=-T$, which cancels the trailing $+T$. Hence
\begin{align*}
g_1\bigl(\wout(s)\mid\wout(s)\bigr)
&=\sum_{i\in[N]}\Dout_i(T)\log\wout_i(s)\\
&\quad-\sum_{t=1}^{T}\log\Bigl(\sum_{i\in[N]}\wout_i(s)\,a_{i,t}\Bigr)\\
&=\ell\bigl(\wout(s)\bigr).
\end{align*}

\emph{Minorization.} Define
\begin{equation}
\Psi_1(t)=\frac{\sum_{i\in[N]}\wout_i\,a_{i,t}}
{\sum_{i\in[N]}\wout_i(s)\,a_{i,t}}>0.
\label{eq:psi1}
\end{equation}
Adding and subtracting
$\sum_{t}\log\sum_{i}\wout_i\,a_{i,t}$ in \eqref{eq:g1} and using
\eqref{eq:psi1} gives
\begin{align}
g_1\bigl(\wout\mid\wout(s)\bigr)
&=\ell(\wout)+\sum_{t=1}^{T}\bigl[\log\Psi_1(t)-\Psi_1(t)+1\bigr].
\label{eq:g1ineq}
\end{align}
The scalar inequality $\log x-x+1\le0$ for all $x>0$ (with equality
iff $x=1$) applied to each $\Psi_1(t)>0$ yields
$g_1(\wout\mid\wout(s))\le\ell(\wout)$ for all $\wout$, i.e.\ $g_1$
minorizes $\ell$. Together with the tangency identity this verifies
the two MM conditions. The argument for $g_2$ is identical after
restricting the inner summation to $i\neq I^{(1)}_t$, which only
removes the source node's term and leaves the inequality
$\log x-x+1\le0$ applicable. \hfill$\blacksquare$

Maximizing \eqref{eq:g1} in $\wout$ (and the analogous $g_2$ in
$\win$) by setting the gradient to zero gives the closed-form
MM updates of the main paper; monotonicity of
the resulting sequence follows from the standard MM sandwich
\begin{align*}
\ell\bigl(\wout(s{+}1)\bigr)
&\ge g_1\bigl(\wout(s{+}1)\mid\wout(s)\bigr)\\
&\ge g_1\bigl(\wout(s)\mid\wout(s)\bigr)=\ell\bigl(\wout(s)\bigr).
\end{align*}

\subsection{Proof of Theorem 5 (MM Convergence)}
The normalization after each closed-form update requires a small
qualification to the usual MM argument. Let $U(v)$ denote the unique
unconstrained maximizer of the surrogate $g(\cdot\mid v)$; this is the
positive closed-form update in the main paper. Let
\[
R(u)=\frac{N u}{\|u\|_1},\qquad M(v)=R\{U(v)\},
\]
so the implemented iteration is $w^{(s+1)}=M(w^{(s)})$. The original
log-likelihood is scale invariant, hence
$\ell\{R(u)\}=\ell(u)$. Minorization and maximization therefore give
\begin{align*}
\ell\{M(v)\}=\ell\{U(v)\}
&\ge g\{U(v)\mid v\}\\
&\ge g(v\mid v)=\ell(v).
\end{align*}
The second inequality is strict unless $U(v)=v$.

It remains to check that normalization does not create a spurious fixed
point. If $M(v)=v$, then $U(v)=cv$ for some $c>0$. Write
$b_i(v)=\sum_t a_{i,t}/\sum_kv_ka_{k,t}$, with the same source masks in
the in-degree case. The surrogate first-order equations at $U(v)=cv$
give $D_i(T)/(cv_i)=b_i(v)$. Multiplying by $v_i$ and summing yields
\[
T/c=\sum_i v_i b_i(v)=T,
\]
because every event contributes one to the masked normalizer. Thus
$c=1$, so $U(v)=v$. By gradient matching, every fixed point of $M$ is a
stationary point of $\ell$.

Under Assumption~1, every terminal count is positive and the normalized
log-likelihood tends to $-\infty$ at the boundary of
$\{w>0:\|w\|_1=N\}$. A non-decreasing MM sequence consequently remains
in a compact upper level set $K$ contained in the positive simplex. The
maps $U$, $R$, and $M$ are continuous on $K$, and the objective values
converge because they are non-decreasing and bounded above. If a
subsequence converges to $w^\star$, continuity gives
$\ell\{M(w^\star)\}=\ell(w^\star)$. Strict improvement away from fixed
points forces $M(w^\star)=w^\star$. Hence every accumulation point is
stationary and, by Theorem~4 of the main paper, equals the unique
normalized optimizer. Compactness then implies convergence of the full
iterate sequence to that optimizer. \hfill$\blacksquare$

\section{Pre-whitening and the Spatial Quasi-Likelihood}
\label{sec:innovation-likelihood}
Let $C_\ell=H^\top\log\widehat{\bm w}_\ell$ be the orthonormal
log-ratio coordinates of the normalized segment weights, with
$H^\top H=I$ and $H^\top\onev=0$. Let
$\bar C=m^{-1}\sum_{\ell=1}^{m}C_\ell$ and
$\widetilde C_\ell=C_\ell-\bar C$. The pooled criterion
\[
\sum_{\ell=2}^{m}\|\widetilde C_\ell-\phi
\widetilde C_{\ell-1}\|^2
\]
gives the pooled AR estimate in the main paper. At the true parameters the
innovations are i.i.d. $\Nrm\{0,\gamma_eR_C(\xi)\}$, where
$R_C(\xi)=H^\top\Sigma_0(\xi)H$.

For the fitted residuals
$E_\ell=\widetilde C_\ell-\widehat\phi\widetilde C_{\ell-1}$, the plug-in Gaussian
log-likelihood, up to constants, is
\[
-\frac{n_e}{2}\log|R_C(\xi)|
-\frac{n_e q}{2}\log\gamma_e
-\frac{1}{2\gamma_e}\sum_{\ell=2}^{m}
E_\ell^\top R_C(\xi)^{-1}E_\ell,
\]
with $n_e=m-1$ and $q=N-1$. Profiling gives
\[
\widehat\gamma_e(\xi)=\frac{1}{n_e q}
\sum_{\ell=2}^{m}E_\ell^\top R_C(\xi)^{-1}E_\ell.
\]
The stationary latent-field variance is
$\widehat\gamma=\widehat\gamma_e/(1-\widehat\phi^2)$, whereas the
spatial range is unchanged. This distinction is important when comparing
innovation-scale estimates with the stationary covariance model in the
main paper. Since $\bar C$, $\phi$, and the weights entering $C_\ell$ are
estimated, the residual vectors are not exactly independent Gaussian
innovations. The displayed likelihood is a conditional quasi-likelihood;
unconditional uncertainty must repeat every estimation stage.

\section{Fenwick-Tree Network Generation}
\label{sec:fenwick}

The diagnostic simulations use the basic categorical generator of
Algorithm~\ref{alg:gen-supp}, which recomputes a
normalized probability vector and performs a linear scan at every
step, costing $O(N)$ per step and $O(TN)$ in total, which is
expensive for large $N$ or long horizons. We accelerate node
selection with a Fenwick tree (binary indexed tree)~\cite{fenwick1994S} via inverse-transform sampling, preserving the
sampling distribution exactly while reducing the per-step cost to
$O(\log N)$.

\begin{algorithm}[h]
\caption{Spatial-network generation (basic categorical sampler)}
\label{alg:gen-supp}
\begin{algorithmic}[1]
\STATE \textbf{Input:} $V=\{1,\dots,N\}$, $T$, $\alpha$, initial $G$,
$\bm{\mu}^{\mathrm{out}},\bm{\mu}^{\mathrm{in}}$,
$(\gamma^{\mathrm{out}},\xi^{\mathrm{out}})$,
$(\gamma^{\mathrm{in}},\xi^{\mathrm{in}})$.
\STATE Sample positions $X$ on $[0,10]^2$;
$\Sigma^{A}\leftarrow\{\gamma^{A}\exp(-\|\bm{x}_i-\bm{x}_j\|/\xi^{A})\}$.
\STATE $w^{A}\sim\mathrm{Lognormal}(\bm{\mu}^{A},\Sigma^{A})$.
\FOR{$t=1$ to $T$}
  \STATE $I^{(1)}_t\!\leftarrow\!\mathrm{Categorical}\!\bigl(
  \wout\!\odot\!(\Dout(t{-}1)+1)^{\alpha}/\|\cdot\|_1\bigr)$
  \STATE set the source mass to zero, then
  $I^{(2)}_t\!\leftarrow\!\mathrm{Categorical}\!\bigl(
  \win\!\odot\!(\Din(t{-}1)+1)^{\alpha}/\|\cdot\|_1\bigr)$
  \STATE $G[I^{(1)}_t,I^{(2)}_t]\leftarrow G[I^{(1)}_t,I^{(2)}_t]+1$
\ENDFOR
\STATE \textbf{return} edge list $\{e_t\}_{t=1}^{T}$.
\end{algorithmic}
\end{algorithm}

\begin{definition}[Prefix sums]
\label{def:prefix}
With node masses $m_i=w_i(D_i+1)^{\alpha}$, let
$S_k=\sum_{i=1}^{k}m_i$ for $k=1,\dots,N$, $S_0:=0$, and
$S_N=\sum_{i=1}^{N}m_i$.
\end{definition}

\begin{lemma}[Interval-length lemma]
\label{lem:interval}
Let $U\sim\mathrm{Unif}(0,S_N)$ and define the random index
$I:=\min\{k:S_k\ge U\}$. Then for every $i\in\{1,\dots,N\}$,
\[
\Prob(I=i)=\Prob(S_{i-1}<U\le S_i)=\frac{m_i}{S_N}.
\]
\end{lemma}

\begin{proof}
Since $U\sim\mathrm{Unif}(0,S_N)$, for any interval
$A\subset[0,S_N]$, $\Prob(U\in A)=|A|/S_N$. The event $\{I=i\}$ is
exactly $\{S_{i-1}<U\le S_i\}$, whose length is
$|S_i-S_{i-1}|=m_i$, giving $\Prob(I=i)=m_i/S_N$.
\end{proof}

\begin{corollary}[Equivalence to categorical sampling]
\label{cor:equiv}
At any step $t$, conditional on the history $\mathcal{F}_{t-1}$, the
index $I_t$ produced by the Fenwick-tree inverse-transform draw
satisfies
\[
\Prob(I_t=i\mid\mathcal{F}_{t-1})
=\frac{m_i(t)}{\sum_{j=1}^{N}m_j(t)},
\]
which coincides with the categorical probability used in
Algorithm~\ref{alg:gen-supp}.
\end{corollary}

\begin{proof}
Immediate from Lemma~\ref{lem:interval}:
$\Prob(I_t=i\mid\mathcal{F}_{t-1})
=\bigl(S_i(t)-S_{i-1}(t)\bigr)/S_N(t)
=m_i(t)/\sum_{j}m_j(t)$, which is exactly the categorical pmf.
\end{proof}

To avoid self-loops the target-node mass $m_{I^{(1)}_t}$ is
temporarily set to zero, which is equivalent to assigning the interval
$(S_{I^{(1)}_t-1},S_{I^{(1)}_t}]$ length zero; Lemma~\ref{lem:interval}
and Corollary~\ref{cor:equiv} continue to hold on
$\{1,\dots,N\}\setminus\{I^{(1)}_t\}$. The Fenwick tree finds
$I=\min\{k:S_k\ge U\}$ by binary lifting over the implicit prefix-sum
array without materializing it, equivalent to a binary search on the
CDF, in $O(\log N)$ time; a single-point update after each new edge
is also $O(\log N)$. Algorithm~\ref{alg:fenwick} summarizes the
procedure.

\begin{algorithm}[t]
\caption{Fenwick-tree spatial-network generation}
\label{alg:fenwick}
\begin{algorithmic}[1]
\STATE \textbf{Input:} $V=\{1,\dots,N\}$, horizon $T$, exponent
$\alpha$, weights $\bm{w}^{\mathrm{out}},\bm{w}^{\mathrm{in}}$.
\STATE Set $\Dout(0)=\Din(0)=0$,
$m^{A}_i(0)=w^{A}_i(D^{A}_i(0)+1)^{\alpha}$,
$A\in\{\mathrm{out},\mathrm{in}\}$; build Fenwick trees on
$\{m^{\mathrm{out}}_i(0)\}$ and $\{m^{\mathrm{in}}_i(0)\}$.
\FOR{$t=1$ to $T$}
  \STATE $S^{\mathrm{out}}\leftarrow\sum_i m^{\mathrm{out}}_i(t-1)$;
  draw $U\sim\mathrm{Unif}(0,S^{\mathrm{out}})$; obtain $I^{(1)}_t$ by
  Fenwick prefix-sum inversion
  \STATE temporarily set the $I^{(1)}_t$ in-mass to $0$; set
  $S^{\mathrm{in}}_{-}=S^{\mathrm{in}}-m^{\mathrm{in}}_{I^{(1)}_t}$,
  draw $V\sim\mathrm{Unif}(0,S^{\mathrm{in}}_{-})$, obtain
  $I^{(2)}_t$ by inversion, and restore the mass
  \STATE
  $\Dout_{I^{(1)}_t}\!\mathrel{+}=\!1$,\;
  $\Din_{I^{(2)}_t}\!\mathrel{+}=\!1$; update
  $m^{\mathrm{out}}_{I^{(1)}_t}$ and $m^{\mathrm{in}}_{I^{(2)}_t}$ in
  the Fenwick trees
\ENDFOR
\STATE \textbf{return} edge list $\{e_t\}_{t=1}^{T}$.
\end{algorithmic}
\end{algorithm}

\paragraph{Complexity}
Each step performs a constant number of $O(\log N)$ prefix-sum
queries, inversions, and single-point updates, so the total cost is
$O(T\log N)$, versus $O(TN)$ for the categorical generator. The
speed-up is by orders of magnitude for moderate-to-large $N$
(Table~\ref{tab:timing}): e.g.\ at $N=1000$ with $30{,}000$ edges the
mean generation time drops from $3165.1$~ms to $4.1$~ms.

\begin{table*}[t]
\caption{Generation time (ms) over $10$ runs: categorical
``Sample'' vs.\ ``Fenwick'' generators.}
\label{tab:timing}
\centering
\footnotesize
\begin{tabular}{@{}llrrrrr@{}}
\toprule
Algorithm & Nodes $N$ & Edges $|E|$ & Min & Mean & Median & Max\\
\midrule
Sample  & 50   & 30000  & 225.2834 & 243.4091 & 231.9460 & 343.7936\\
Fenwick & 50   & 30000  & 2.6693   & 2.9426   & 2.8767   & 3.6472\\
Sample  & 50   & 50000  & 372.5650 & 383.3666 & 383.1850 & 391.7366\\
Fenwick & 50   & 50000  & 4.4940   & 4.9647   & 4.7005   & 7.4118\\
Sample  & 50   & 100000 & 750.7992 & 788.4982 & 763.8286 & 981.2024\\
Fenwick & 50   & 100000 & 8.9907   & 10.3229  & 9.4652   & 18.3140\\
Sample  & 100  & 30000  & 302.8406 & 352.7672 & 318.7164 & 570.9076\\
Fenwick & 100  & 30000  & 3.0014   & 3.1387   & 3.0712   & 3.5138\\
Sample  & 100  & 50000  & 511.3196 & 550.4037 & 523.9248 & 774.9083\\
Fenwick & 100  & 50000  & 4.9517   & 5.2704   & 5.1831   & 5.8662\\
Sample  & 100  & 100000 & 1031.0781& 1081.7682& 1051.4940& 1304.7946\\
Fenwick & 100  & 100000 & 10.0438  & 10.4392  & 10.2312  & 12.0284\\
Sample  & 300  & 30000  & 703.3597 & 773.1725 & 715.3936 & 1105.1344\\
Fenwick & 300  & 30000  & 3.4771   & 4.0965   & 3.6928   & 5.8707\\
Sample  & 300  & 50000  & 1162.3771& 1245.3561& 1176.2700& 1466.6472\\
Fenwick & 300  & 50000  & 5.8081   & 6.0837   & 5.9988   & 6.5744\\
Sample  & 300  & 100000 & 2344.5202& 2429.4927& 2363.0491& 2880.3861\\
Fenwick & 300  & 100000 & 11.5492  & 12.0827  & 11.7371  & 14.8995\\
Sample  & 500  & 30000  & 1213.2895& 1377.8721& 1319.4865& 1785.6402\\
Fenwick & 500  & 30000  & 3.6893   & 4.1103   & 3.8266   & 5.4324\\
Sample  & 500  & 50000  & 2000.8517& 2094.6660& 2018.3291& 2556.4067\\
Fenwick & 500  & 50000  & 6.1826   & 6.3159   & 6.2856   & 6.4884\\
Sample  & 1000 & 10000  & 1021.8911& 1067.1172& 1029.8842& 1281.3651\\
Fenwick & 1000 & 10000  & 1.3554   & 1.4041   & 1.4049   & 1.4703\\
Sample  & 1000 & 30000  & 3057.0497& 3165.1426& 3140.5722& 3353.8776\\
Fenwick & 1000 & 30000  & 4.0474   & 4.1408   & 4.1291   & 4.3095\\
\bottomrule
\end{tabular}
\end{table*}

\section{Independent-Weight Baselines}
\label{sec:b2}

We compare three PA-limit fields on the same $N=80$ locations
(Fig.~\ref{fig:ablation-supp}); for reference we label
them B2 (i.i.d.\ lognormal weights), B3 (spatially varying marginals
with independent weights), and B4 (Gaussian-process (GP) correlated
lognormal weights),
with B1 denoting standard PA. This section gives the precise
specification of the two independent baselines (B2, B3).

Baseline B2 isolates the effect of the spatial correlation by removing
it while keeping heavy tails. It uses exactly the proposed setup of
B4, that is, the same $N=80$ locations, the same zero mean, and the same
generator and F-madogram and madogram-coefficient estimators, with the
single change that the exponential covariance of the main paper is
replaced by a diagonal one,
\begin{equation}
\log\bm{w}(\bm{x})\sim\Nrm\!\bigl(\bm{\mu}(\bm{x}),\,\gamma\,I_N\bigr),
\qquad\gamma=1.0,
\label{eq:b2}
\end{equation}
i.e.\ $\Sigma_{ij}=\gamma\,\delta_{ij}$, equivalently the
exponential covariance with $\xi\to0^{+}$. The marginals are
unchanged (lognormal, hence heavy-tailed), so any change in the
diagnostics relative to B4 is attributable solely to the loss of
spatial correlation. Empirically the F-madogram and the madogram
coefficient for \eqref{eq:b2} are flat in the lag distance,
confirming that heavy-tailedness alone does not generate
distance-decaying finite-level dependence. Baseline B3 keeps independent
(diagonal) weights but adds a smooth, deterministic spatial mean
surface to the log-weights,
\begin{align}
\log\bm{w}(\bm{x})
&\sim\Nrm\!\bigl(\bm{\mu}_0,\,\gamma I_N\bigr),\nonumber\\
\{\bm{\mu}_0\}_i
&=0.25(x_{i1}-\bar{x}_1)-0.20(x_{i2}-\bar{x}_2).
\label{eq:b3}
\end{align}
so the marginals vary smoothly across space while the weights remain
spatially independent. This mimics a spatially varying mean without a
joint spatial-dependence structure; its diagnostics are likewise
flat, confirming that spatially varying marginals alone do not
generate distance-decaying finite-level dependence. Baseline B1 (standard
PA) has no node coordinates, so the diagnostics are undefined and that
comparison is conceptual.

\begin{figure}[t]\centering
\includegraphics[width=\columnwidth]{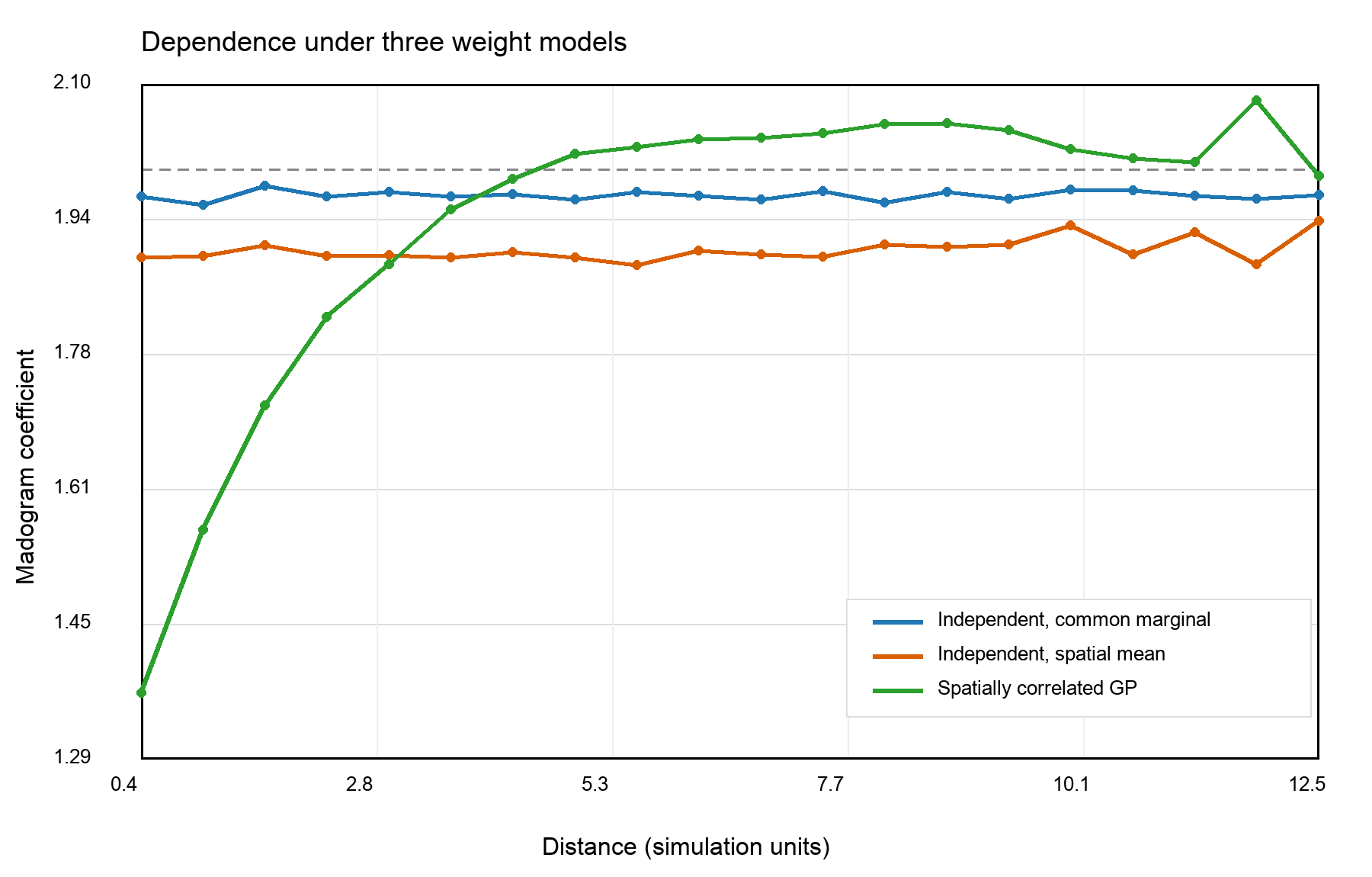}
\caption{Model comparison at the PA limit. Heavy-tailed independent weights
and spatially varying but independent weights produce nearly flat dependence
profiles, whereas GP-correlated weights yield dependence that decays with
distance. The grey line marks the independence reference $\theta_F=2$.}
\label{fig:ablation-supp}
\end{figure}

\section{Weight Recovery and PA-Exponent Identification}
\label{sec:weight-recovery}
We assess the estimation steps at a scale large enough to show the trends but
small enough to reproduce easily. For weight recovery we use
$N\in\{50,75,100\}$, $T\in\{5000,10000,20000,50000\}$, six independent
replications per setting, and the mean squared error
$N^{-1}\sum_i(\widehat w_i-w_i)^2$. Figure~\ref{fig:weight-mse} shows monotone
improvement as $T$ increases, with the larger network requiring more edges for
the same accuracy; this supports using the MM weights as inputs to the spatial
covariance step.

\begin{figure}[t]\centering
\includegraphics[width=\columnwidth]{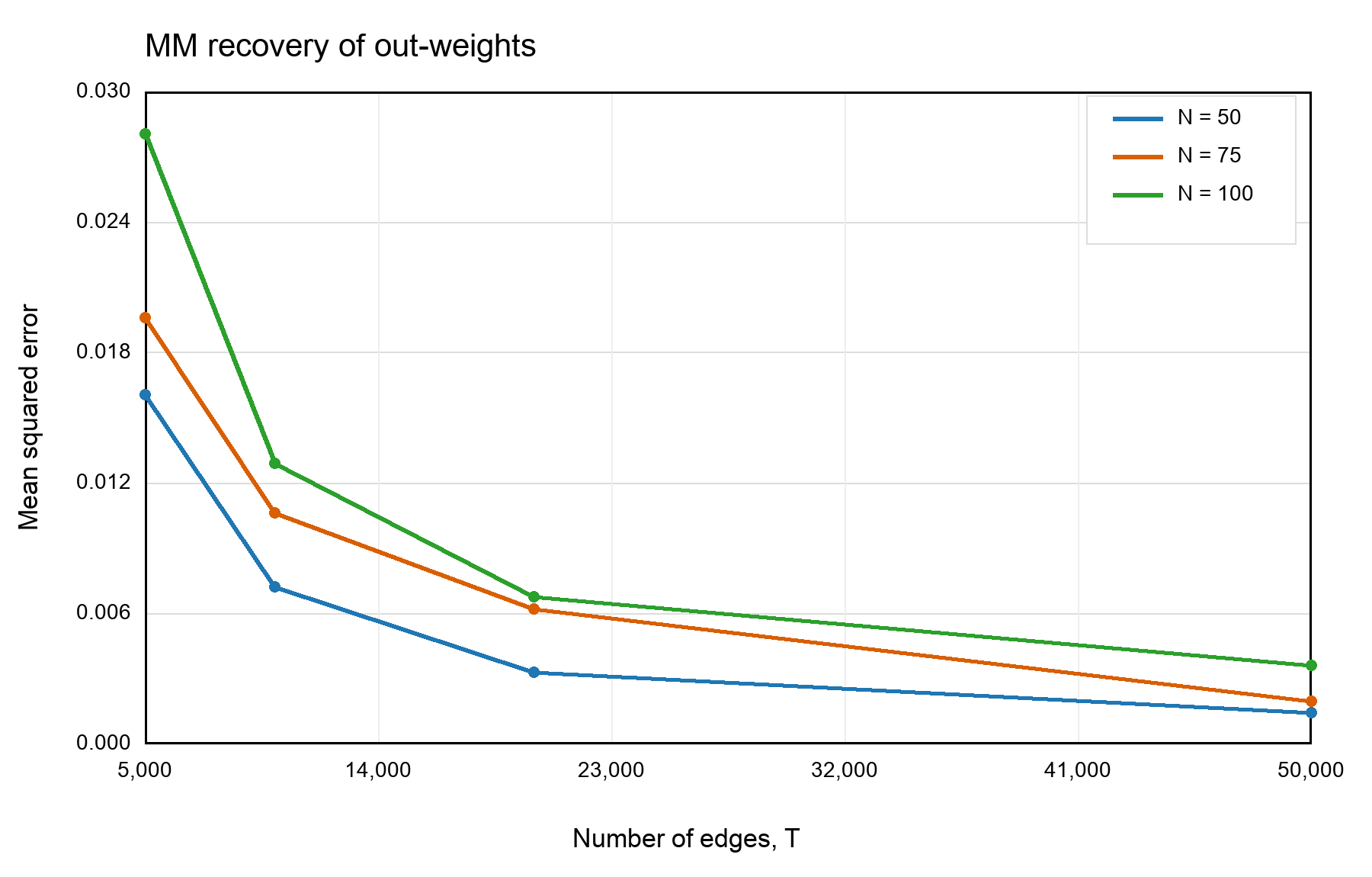}
\caption{MM out-weight recovery for $N\in\{50,75,100\}$ and increasing edge
horizons. MSE decreases with $T$, and larger networks are harder at the same
edge count.}
\label{fig:weight-mse}
\end{figure}

For the PA exponent, Figure~\ref{fig:alpha} profiles $\ell_p(\alpha)$ at
$N=50$ for $T=20000$ and $T=50000$. The profiles are smooth and single-peaked,
but the maximizers, about $0.33$ and $0.43$, remain below the true value
$\alpha_0=0.5$. Once the degree proportions approach their stationary
counterpart, the data identify the powered weight vector more strongly than
$\alpha$ itself, so $\widehat\alpha$ is a useful profile summary rather than an
accurately recovered parameter at these horizons.

\begin{figure}[t]\centering
\includegraphics[width=\columnwidth]{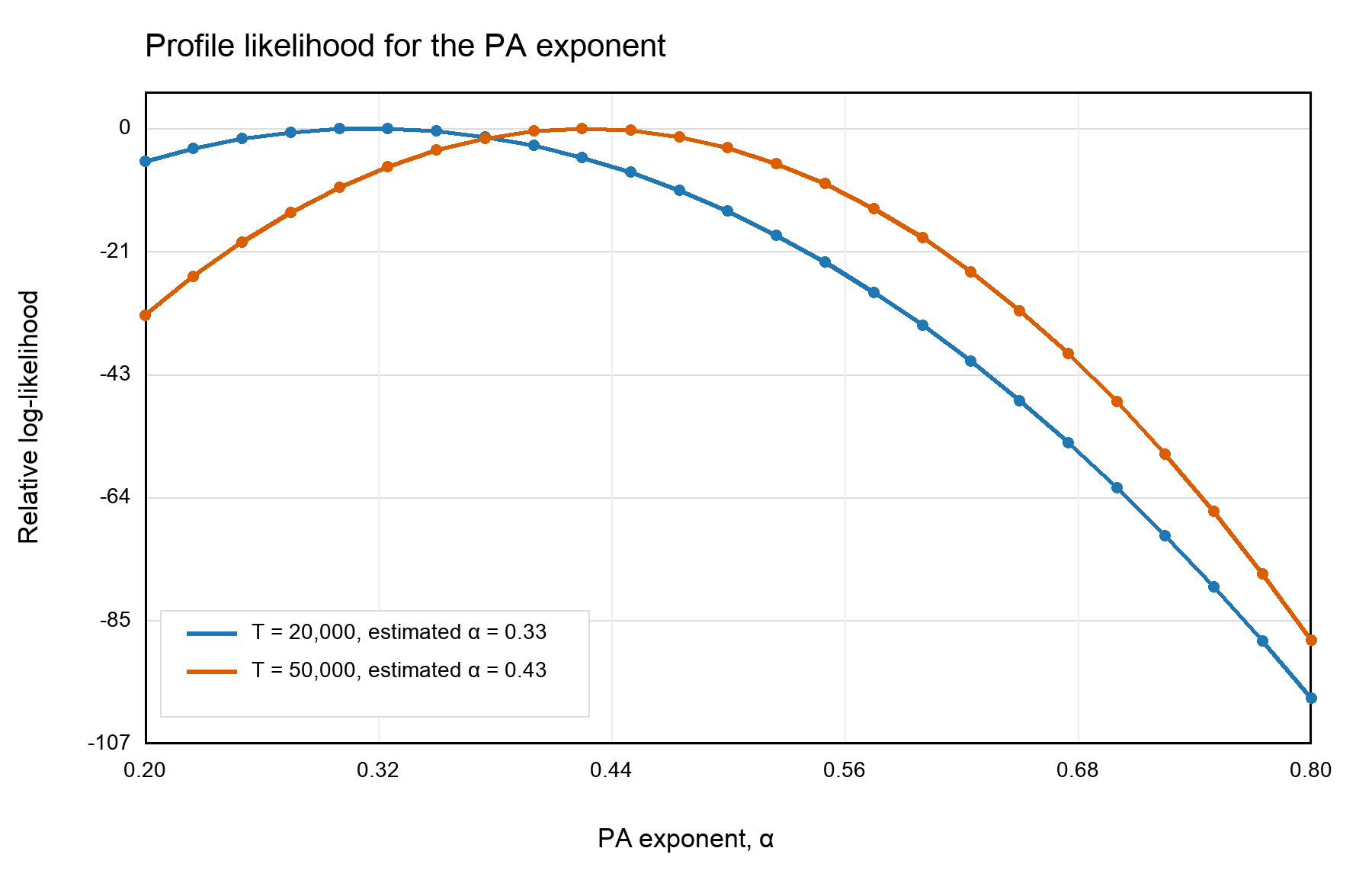}
\caption{Profile likelihood for $\alpha$ at $N=50$. The profile is smooth but
conservative at moderate horizons, moving from $\widehat\alpha\approx0.33$ at
$T=20000$ to $\widehat\alpha\approx0.43$ at $T=50000$.}
\label{fig:alpha}
\end{figure}

\section{Transmission to the Degree Structure}
\label{sec:degdep}
Beyond the weight-level comparison, we verify that the GP spatial
dependence of the proposed field (B4) is inherited by the
\emph{generated network}. This complements the transmission experiment
of Section~VI-B of the main paper by generating the full directed
multigraph, that is, both source (out) and target (in) selection with
self-loop exclusion, and separating the out- and in-degree channels.
Using the canonical dependence-diagnostic setup ($N=80$, zero mean,
$(\gamma^{A},\xi^{A})=(1.0,2.0)$, $\alpha=0.5$, $A\in\{\mathrm{out},
\mathrm{in}\}$, $30$ replicates of $T=20{,}000$ edges),
Fig.~\ref{fig:degdep} reports the madogram dependence coefficient and F-madogram
of the simulated out- and in-degrees against lag distance. Both
estimators rise from strong short-range dependence toward the
independence reference and the out- and in-degree diagnostics nearly
coincide. This confirms that finite-level spatial dependence is
transmitted from the latent weights through the preferential-attachment
aggregation to the observable out-degree proportions. The in-degree curve is
a simulation result for the coupled target process, consistent with the
implicit in-degree-proportion limit rather than the powered out-weight formula. Its
near agreement with the out-channel in this setting illustrates the
small-excluded-mass approximation.

\begin{figure}[t]\centering
\includegraphics[width=\columnwidth]{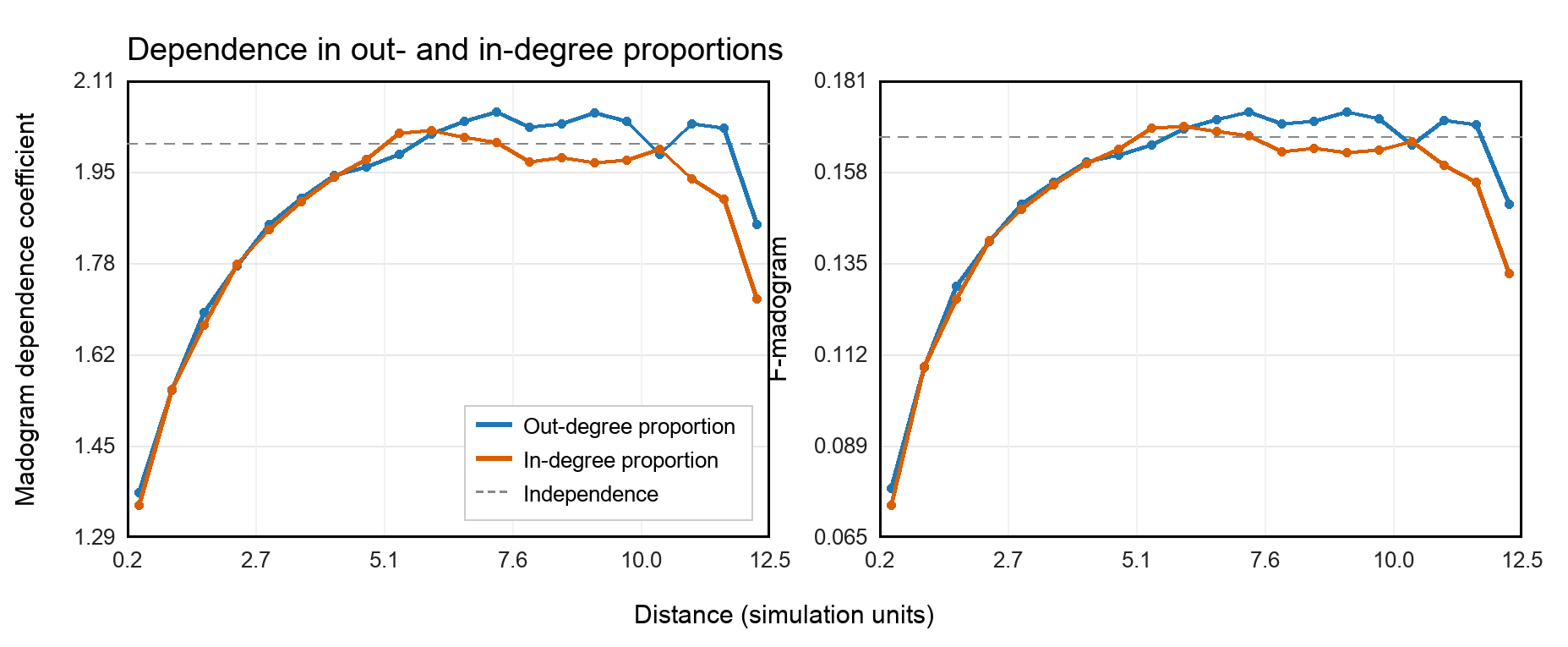}
\caption{Spatial dependence of the generated out- and in-degree proportions:
madogram coefficient (left) and F-madogram (right) vs lag
distance in simulation units. Full directed PA multigraph, $N=80$, lognormal field with
zero mean, $\gamma=1.0$, $\xi=2.0$, $\alpha=0.5$; grey lines mark the
independence reference values $\theta_F=2$ and $v_F=1/6$.}
\label{fig:degdep}\end{figure}

\section{Full Analysis of Dataset~II (Other Division)}
\label{sec:od}

Dataset~II (the Other Division, OD) comprises the $61$ major airports
outside the Eastern Division, with $8{,}981{,}748$ flight records over
2015--2019 whose endpoints both lie in OD, processed into three
within-month observational units as in the main
paper.

\subsection{Spatial Dependence in Airport Activity}
Fig.~\ref{fig:od_dep} reports the binned madogram coefficient and
F-madogram for the $61$ OD airports against geodesic distance, computed
on the raw degrees as in the main paper. The qualitative behavior
matches the ED network---both curves increase monotonically with
distance and plateau below the independence reference---reflecting the same
two ingredients: short-range spatial dependence and variation in observed
segment volume. The OD points reach somewhat
higher values at the longest lags (approaching $1.9$ near $3000$~km),
consistent with a slightly larger fitted spatial range
(Table~\ref{tab:regionS}), though, as for ED, the long-range level is
governed mainly by the common mode rather than by spatial proximity.
Both fitted spatial components are short-range, although the pre-whitened
OD estimate is longer than ED's (Section~\ref{sec:prewhite-diagnostics}).

\begin{figure}[t]\centering
\includegraphics[width=\columnwidth]{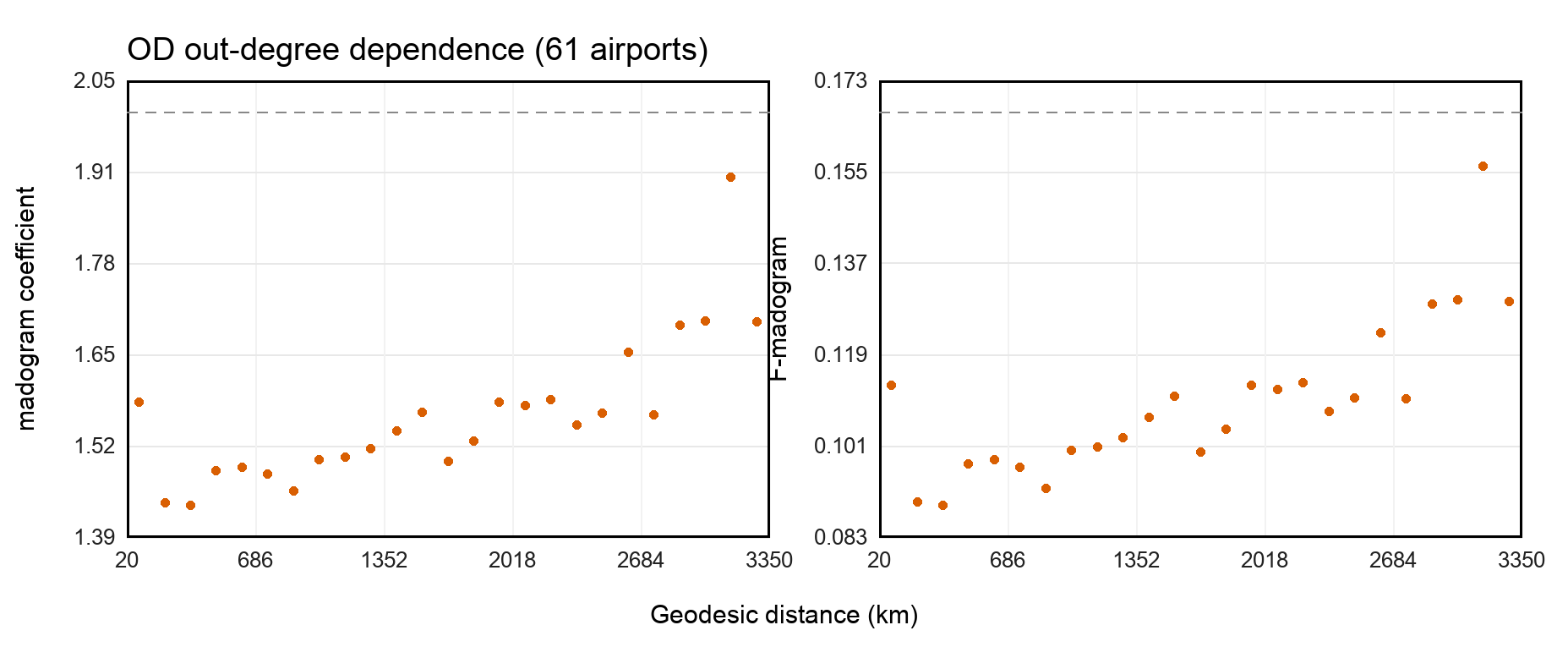}
\caption{Empirical madogram coefficient (left) and F-madogram
(right) vs geodesic distance (km) for the 61 OD airports.}
\label{fig:od_dep}\end{figure}

\subsection{Observed-Volume Reconstruction}
Fig.~\ref{fig:od_modelfit} repeats the observed-volume decomposition of
the main paper for the OD network. The first principal component of the
segment-to-segment log out-degrees accounts for $64.8\%$ of
the variance and correlates $0.77$ with the per-segment total volume,
a weaker common volume mode than in ED ($82.6\%$ and $0.99$).
Working with degree proportions removes the direct volume scale and lifts the madogram
coefficient to the independence reference, whereas recombining the degree proportions
with independently resampled volumes reproduces the raw-degree common mode, as
in the ED reconstruction. The OD common mode is thus less volume-dominated and relatively
more spatial than the ED common mode, consistent with the slightly larger
fitted spatial range of the OD network.

\begin{figure}[t]\centering
\includegraphics[width=\columnwidth]{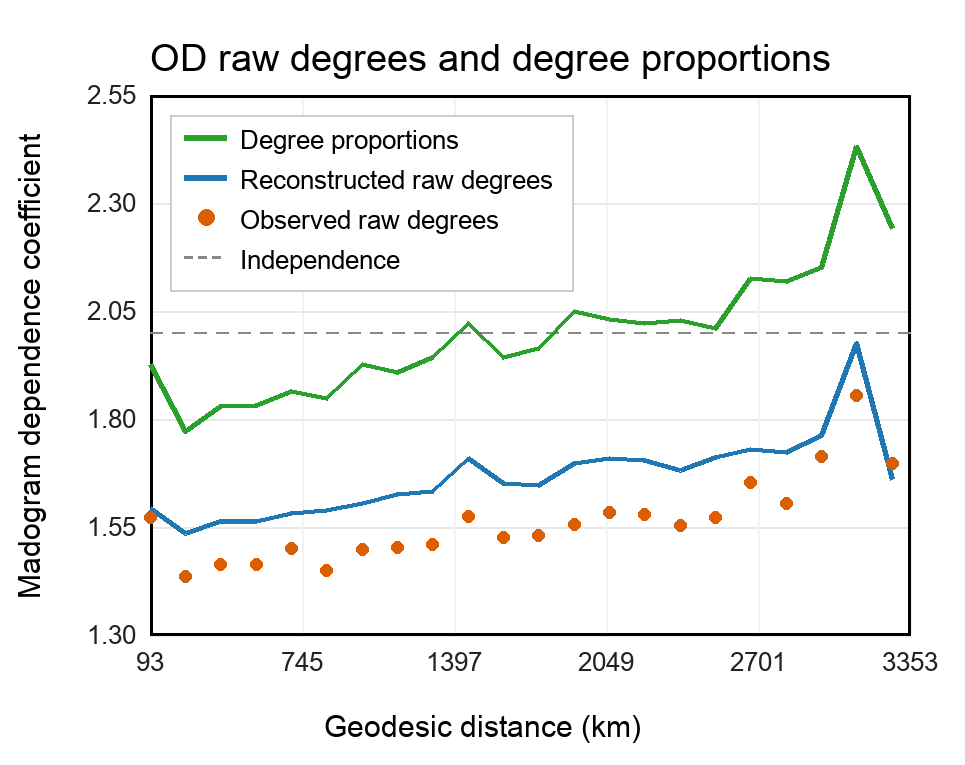}
\caption{Observed-volume reconstruction for the OD network. Markers: madogram
coefficient of the raw out-degrees. Blue: the
semiparametric reconstruction (degree proportions recombined with independently
resampled segment volumes), which reproduces the long-distance dependence. Green: the degree proportions, which rise toward the independence reference $\theta_F=2$.}
\label{fig:od_modelfit}\end{figure}

\subsection{Parameter Inference}
At the reference exponent $\alpha=0.2$, the exact large-horizon out-degree
proportion inversion is transformed to orthonormal log-ratio coordinates and pre-whitened
as in Section~\ref{sec:innovation-likelihood}. The in-weights are recovered
by the exact no-self-loop inversion in Proposition~3 of the main paper and
are then processed identically.
The OD estimates are
$\widehat\phi=0.917$ and $0.920$ for the out- and in-fields, with
innovation variances $0.00385$ and $0.00383$. Converting to stationary
variance gives $0.0242$ and $0.0249$, while the spatial e-folding ranges
are $109.8$ and $109.5$~km, with conditional profile intervals
$[103.6,116.4]$ and $[103.3,116.0]$~km. These intervals condition on the
constructed weights and fitted autoregressive (AR) parameters. The median
residual autocorrelation functions (ACFs) are
$0.038$ and $0.034$, supporting the plug-in innovation approximation. Table
\ref{tab:regionS} reproduces the cross-region comparison.

\begin{table}[t]
\caption{Pre-whitened spatial-component estimates by region at
$\alpha=0.2$ (cf.\ Table~II of the main paper). $\gamma$ is stationary.}
\label{tab:regionS}
\centering
\scriptsize
\setlength{\tabcolsep}{2.5pt}
\begin{tabular}{@{}lcccccc@{}}
\toprule
Region & $\phi^{\rm out}$ & $\gamma^{\rm out}$ & $\xi^{\rm out}$ &
$\phi^{\rm in}$ & $\gamma^{\rm in}$ & $\xi^{\rm in}$\\
\midrule
ED & 0.934 & 0.0274 & 0.736 & 0.934 & 0.0281 & 0.729\\
OD & 0.917 & 0.0242 & 1.098 & 0.920 & 0.0249 & 1.095\\
\bottomrule
\end{tabular}
\end{table}

\subsection{Cross-Region Comparison}
The two regions share the same qualitative structure: short-range spatial
variation superposed on observed segment-volume variation. The
pre-whitened spatial e-folding ranges are approximately $74$~km for ED and
$110$~km for OD. Within each region the
out-/in-parameters remain close, descriptively reflecting the structural
symmetry of the U.S. domestic network. A finer regional comparison would require
either a richer spatiotemporal model or a longer record.

\section{Temporal Pre-whitening Diagnostics}
\label{sec:prewhite-diagnostics}

Table~\ref{tab:prewhite} reports the complete two-stage estimates. The
large AR coefficients reflect persistent flight schedules, while the
residual ACFs are close to zero after filtering. The conditional
95\% profile intervals are $[69.8,77.6]$ and $[69.2,76.9]$~km for ED,
and $[103.6,116.4]$ and $[103.3,116.0]$~km for OD; they do not include
uncertainty from weight construction or AR fitting. The Gaussian
innovation model applies to the latent residual process; fitted residuals
are used through the conditional quasi-likelihood. The innovation variance
$\gamma_e$ is much smaller than the
stationary variance because
$\gamma_e=(1-\phi^2)\gamma$. 

\begin{table}[t]
\caption{Temporal pre-whitening and spatial innovation diagnostics.
Spatial ranges are in km.}
\label{tab:prewhite}
\centering
\scriptsize
\begin{tabular}{@{}llrrrrr@{}}
\toprule
Region & Field & $\widehat\phi$ & $\widehat\gamma_e$ &
$\widehat\gamma$ & $\widehat\xi$ & Median ACF(1)\\
\midrule
ED & Out & 0.934 & 0.00350 & 0.0274 & 73.6 & $-0.051$\\
ED & In  & 0.934 & 0.00358 & 0.0281 & 72.9 & $-0.054$\\
OD & Out & 0.917 & 0.00385 & 0.0242 & 109.8 & 0.038\\
OD & In  & 0.920 & 0.00383 & 0.0249 & 109.5 & 0.034\\
\bottomrule
\end{tabular}
\end{table}

\section{Cross-Checks on the PA Exponent}
\label{sec:alpha-checks}

The main paper uses $\alpha=0.2$ as an empirically informed reference for
the airline analysis and establishes asymptotic range invariance for compatible
transformed out-weights. As an additional
check on the strength of preferential attachment, we estimate $\alpha$
by two methods that use the edge dynamics directly rather than the
stationary profile likelihood, and we calibrate each against synthetic
networks with a known exponent and the empirical hub heterogeneity.

\subsection{Methods}
The first method is the attachment-rate regression of~\cite{jeongS}. For a target node currently at degree $k$, the
per-opportunity attachment rate is estimated as the number of selections
of a degree-$k$ node divided by the node-events spent at degree $k$, and
$\alpha$ is the slope of $\log A(k)$ on $\log(k+1)$. This estimator
ignores node fitness, so it is biased upward when fitness is
heterogeneous: on synthetic networks with the empirical hub heterogeneity
and a true $\alpha_0=0.5$ it returns about $0.85$.

The second method is a PAFit-style joint estimator~\cite{phamS}. The
attachment function $A(k)$ and the node fitness are estimated jointly by
alternating multiplicative updates, which removes most of the fitness
confound. A curvature penalty on $\log A(k)$ over geometric degree bins
regularizes the sparse high-degree tail. Degrees are reset per segment,
matching the per-segment weight estimation, and the fitness is shared
across segments.

\subsection{Calibration and Results}
On synthetic data with the empirical hub heterogeneity, the PAFit-style
estimator recovers a known $\alpha_0=0.5$ as about $0.56$, far closer
than the attachment-rate value of $0.85$. The curvature penalty lowers
the level by roughly $0.02$ without materially changing sensitivity to
the true exponent. Because the attachment likelihood is conditional on
the segment edge count, observed volume does not require an additional
latent factor in this cross-check.

Table~\ref{tab:alpha-calib} reports the PAFit-style calibration across a
range of true exponents, without and with the penalty, together with the
real ED in-degree reading. The estimate increases with $\alpha_0$ but
with a strongly attenuated slope of about $0.25$ rather than the ideal
slope of one, and it saturates in the upper range, so large exponents are
only weakly distinguished. The real reading sits at the level of the
$\alpha_0=0.2$ row.

\begin{table}[t]
\caption{PAFit-style estimate $\widehat\alpha$ versus the true exponent
$\alpha_0$ on synthetic networks with the empirical hub heterogeneity,
without and with the curvature penalty, and the real ED in-degree
reading.}
\label{tab:alpha-calib}
\centering
\begin{tabular}{@{}lcc@{}}
\toprule
True $\alpha_0$ & $\widehat\alpha$ (no penalty) & $\widehat\alpha$ (penalty)\\
\midrule
$0.0$ & 0.433 & 0.408\\
$0.2$ & 0.489 & 0.466\\
$0.4$ & 0.543 & 0.520\\
$0.5$ & 0.564 & 0.540\\
$0.6$ & 0.581 & 0.556\\
$0.7$ & 0.606 & 0.582\\
$0.8$ & 0.671 & 0.648\\
\midrule
real ED in-degree & 0.488 & 0.455\\
\bottomrule
\end{tabular}
\end{table}

\subsection{Interpretation}
Both estimators, once calibrated, favor a small exponent of about
$0.15$ to $0.2$ and provide no support for strong preferential
attachment. The saturation of the calibration in the upper range means
the data weakly distinguish large exponents, which is the same weak
identifiability seen in the profile likelihood. These cross-checks are
exploratory: the day-aggregated records do not resolve the within-segment
edge order, which biases the apparent attachment downward, and the
PAFit-style estimator retains a residual fitness bias that the
calibration accounts for. They are consistent with the weak-identification
analysis of the main paper. The invariance proposition concerns the
asymptotic range; finite-horizon variances and rankings can remain
sensitive to the chosen exponent. 

\subsection{Sensitivity of the Spatial Estimates to $\alpha$}
\label{sec:alpha-sens}
Proposition~2 of the main paper predicts that, in the large-horizon
out-channel inversion $\log w^{\mathrm{out}}_{\ell i}=(1-\alpha)
\log(D^{\mathrm{out}}_{\ell i}/T_\ell)+c^{\mathrm{out}}_\ell$, changing
$\alpha$ rescales every log-ratio coordinate by the common factor $(1-\alpha)$.
Both the pooled AR estimator and the spatial profile objective are
invariant to a common positive scaling of the coordinates, so $\widehat\phi$
and the e-folding range $\widehat\xi$ are unchanged, while the stationary
variance $\widehat\gamma$ rescales by $\{(1-\alpha')/(1-\alpha)\}^2$.
For the exact in-degree inverse, $\widehat v$ in Proposition~3 depends only
on the observed out- and in-degree proportions, but
$\log\widehat w_j^{\mathrm{in}}(\alpha)
=\widehat v_j-\alpha\log p_j+c(\alpha)$ is not generally a common
rescaling. Its range must therefore be re-estimated.

Table~\ref{tab:alpha-sens} compares the reference $\alpha=0.2$ with
$\alpha=0.5$. The out-range and $\widehat\phi$ are invariant, while the
out-variance scales by $\{(1-0.5)/(1-0.2)\}^2=0.390625$. Under the exact
inversion, the ED in-range changes from $72.9$ to $72.2$~km and the OD
in-range from $109.5$ to $108.5$~km. Thus the exact algebraic invariance is
specific to the out channel, but the empirical in-range is also insensitive
over these two values. The variance remains the parameter most affected by
$\alpha$.

\begin{table}[t]
\caption{Spatial-component estimates at the reference $\alpha=0.2$ and the
sensitivity value $\alpha=0.5$. The out channel uses the explicit inversion;
the in channel uses the exact no-self-loop inversion. Ranges are in km.}
\label{tab:alpha-sens}
\centering
\scriptsize
\setlength{\tabcolsep}{3pt}
\begin{tabular}{@{}llcccccc@{}}
\toprule
& & \multicolumn{3}{c}{$\alpha=0.2$} & \multicolumn{3}{c}{$\alpha=0.5$}\\
\cmidrule(lr){3-5}\cmidrule(lr){6-8}
Region & Field & $\widehat\phi$ & $\widehat\gamma$ & $\widehat\xi$ &
$\widehat\phi$ & $\widehat\gamma$ & $\widehat\xi$\\
\midrule
ED & Out & 0.934 & 0.0274 & 73.6  & 0.934 & 0.0107 & 73.6\\
ED & In  & 0.934 & 0.0281 & 72.9  & 0.934 & 0.0111 & 72.2\\
OD & Out & 0.917 & 0.0242 & 109.8 & 0.917 & 0.0095 & 109.8\\
OD & In  & 0.920 & 0.0249 & 109.5 & 0.920 & 0.0098 & 108.5\\
\bottomrule
\end{tabular}
\end{table}

\end{document}